\documentclass[10pt,journal,epsfig]{IEEEtran}

\usepackage[dvips]{graphicx}
\usepackage{graphicx}
\usepackage{amssymb}
\usepackage{cite}
\usepackage{subfigure}
\usepackage{amsthm}
\usepackage{amsmath}
\usepackage{algorithm}
\usepackage{algorithmic}
\usepackage{multirow}
\usepackage{xcolor}

\begin{document}

\title{A Variational Bayesian Inference-Inspired Unrolled Deep Network for MIMO Detection}

\author{Qian Wan, Jun Fang, ~\IEEEmembership{Senior Member}, Yinsen Huang, Huiping
Duan, and Hongbin Li, ~\IEEEmembership{Fellow,~IEEE}
\thanks{Qian Wan, Jun Fang and Yinsen Huang are with the National Key Laboratory
of Science and Technology on Communications, University of
Electronic Science and Technology of China, Chengdu 611731, China,
Email: JunFang@uestc.edu.cn}
\thanks{Huiping Duan is with the School of Information and Communications Engineering,
University of Electronic Science and Technology of China, Chengdu
611731, China, Email: huipingduan@uestc.edu.cn}
\thanks{Hongbin Li is
with the Department of Electrical and Computer Engineering,
Stevens Institute of Technology, Hoboken, NJ 07030, USA, E-mail:
Hongbin.Li@stevens.edu}
\thanks{This work was supported in part by the National Science
Foundation of China under Grants 61829103 and 61871421.}

\thanks{\textcopyright 2022 IEEE. Personal use of this material is permitted.
  Permission from IEEE must be obtained for all other uses, in any current or future
  media, including reprinting/republishing this material for advertising or promotional
  purposes, creating new collective works, for resale or redistribution to servers or
  lists, or reuse of any copyrighted component of this work in other works.}
}

\maketitle

\begin{abstract}
The great success of deep learning (DL) has inspired researchers
to develop more accurate and efficient symbol detectors for
multi-input multi-output (MIMO) systems. Existing DL-based MIMO
detectors, however, suffer several drawbacks. To address these
issues, in this paper, we develop a model-driven DL detector based
on variational Bayesian inference. Specifically, the proposed
unrolled DL architecture is inspired by an inverse-free
variational Bayesian learning framework which circumvents matrix
inversion via maximizing a relaxed evidence lower bound. Two
networks are respectively developed for independent and
identically distributed (i.i.d.) Gaussian channels and arbitrarily
correlated channels. The proposed networks, referred to as VBINet,
have only a few learnable parameters and thus can be efficiently
trained with a moderate amount of training samples. The proposed
VBINet-based detectors can work in both offline and online
training modes. An important advantage of our proposed networks
over state-of-the-art MIMO detection networks such as OAMPNet and
MMNet is that the VBINet can automatically learn the noise
variance from data, thus yielding a significant performance
improvement over the OAMPNet and MMNet in the presence of noise
variance uncertainty. Simulation results show that the proposed
VBINet-based detectors achieve competitive performance for both
i.i.d. Gaussian and realistic 3GPP MIMO channels.
\end{abstract}


\section{Introduction}
Massive multiple-input multiple-output (MIMO) is a promising
technology for next-generation wireless communication systems
\cite{ZhangCui05,ShahmansooriGarcia17,LuLi14}. In massive MIMO
systems, the base station (BS), equipped with tens or hundreds of
antennas, simultaneously serves a smaller number of single-antenna
users sharing the same time-frequency resource
\cite{NgoLarsson13}. The large number of antennas at the BS helps
achieve almost perfect inter-user interference cancelation and has
the potential to enhance the spectrum efficiency by orders of
magnitude. Nevertheless, large-scale MIMO poses a significant
challenge for optimal signal detection. To realize the full
potential of massive MIMO systems, it is of significance to
develop advanced signal detection algorithms that can strike a
good balance between performance and complexity \cite{KongXia16}.

Signal detection in the MIMO framework has been extensively
investigated over the past decades. It is known that the maximum
likelihood (ML) detector is an optimal detector but has a
computational complexity that scales exponentially with the number
of unknown variables. Linear detectors such as zero-forcing (ZF)
and linear minimum mean-squared error (LMMSE) have a relatively
low computational complexity, but they suffer a considerable
performance degradation as compared with the ML detector.
Recently, Bayesian inference or optimization-based iterative
detectors, e.g.
\cite{tipping01,DonohoMaleki10,ChenWipf16,WuKuang14,Minka01,jaldenOttersten08},
were proposed. Empirical results show that these iterative schemes
such as the approximate message passing (AMP) and expectation
propagation (EP)-based detectors can achieve excellent performance
with moderate complexity \cite{WuKuang14,Minka01}. Nevertheless,
the AMP-based or EP-based methods utilize a series of
approximation techniques to facilitate the Bayesian inference.
When the channel matrix has a small size or is ill-conditioned,
the approximation may not hold valid, in which case these methods
may suffer severe performance degradation.

More recently, the great success of deep learning (DL) in a
variety of learning tasks has inspired researchers to use it as an
effective tool to devise more accurate and efficient detection and
estimation schemes
\cite{WangWenWang17,WeiZhao20,YeLi17,BaiChen20,HuCai20,MaGao21,HuangLiu21}.
The application of purely data-driven DL-based schemes, however,
faces the challenges of a long training time as a result of the
complicated network structure and a large number of learnable
parameters \cite{BaekKwak19,FarsadGoldsmith18}. To address this
issue, model-driven DL-based schemes were recently proposed for
MIMO detection \cite{UnShao19,HeWen20,KhaniAlizadeh20}. The main
idea of model-driven DL detector is to unfold the iterative
algorithm as a series of neural network layers with some learnable
parameters. Specifically, DetNet seems to be one of the earliest
works to study the problem of MIMO detection by unfolding the
projected gradient descent algorithm \cite{SamuelDiskin17}. The
number of learnable parameters for DetNet, however, is still
large. Thus it needs a large number of samples to train the neural
network. As an extension of DetNet, WeSNet aims to reduce the
model size by introducing a sparsity-inducing regularization
constraint in conjunction with trainable weight-scaling functions
\cite{MohammadMasouros20}. In order to further reduce the number
of learnable parameters of the network, OAMPNet, a neural
architecture inspired by the orthogonal AMP algorithm, was
proposed \cite{HeWen20}. The OAMPNet has only four learnable
parameters per layer, and thus is easy to train. Compared with
DetNet, OAMPNet can achieve competitive performance with much
fewer training samples. Both DetNet and OAMPNet are trained in an
offline mode, i.e. they aim to learn a single detector during
training for a family of channel matrices. Recently, another work
\cite{KhaniAlizadeh20} examined the MIMO detection problem from an
online training perspective. The proposed MMNet can achieve a
considerable performance improvement over OAMPNet when sufficient
online training is allowed. Nevertheless, online training means
that the model parameters are trained and tested for each
realization of the channel. When the test channel sample is
different from the training channel sample, MMNet would suffer a
substantial amount of performance degradation. It should be noted
that both OAMPNet and MMNet require the knowledge of the noise
variance for symbol detection. In practice, however, the prior
information about the noise variance is usually unavailable and
inaccurate estimation may result in substantial performance loss.

Recently, the information bottleneck (IB) theory attracts much
attention for providing an information theoretic view of deep
learning networks \cite{AguerriZaidi21,ZaidiAguerri20}. The IB
approach has found applications in a variety of learning problems.
It can also be applied to solve the MIMO detection problem.
Consider extracting the relevant information that some signal $X$
provides about another one $Y$ that is of interest. IB formulates
the problem of finding a representation $U$ that is maximally
informative about $Y$, while being minimally informative about $X$
\cite{AguerriZaidi21}. Accordingly, the optimal mapping of the
data $X$ to $U$ is found by solving a Lagrangian formulation. This
problem can be solved by approximating its variational bound by
parameterizing the encoder, decoder and prior distributions with
deep neural networks. Nevertheless, when applied to MIMO
detection, the IB approach needs a large number of learnable
parameters. Therefore it still faces the challenge encountered by
purely data-driven DL-based schemes.

In this paper, we propose a variational Bayesian
inference-inspired unrolled deep network for MIMO detection. Our
proposed deep learning architecture is mainly inspired by the
inverse-free Bayesian learning framework \cite{DuanYang17}, where
a fast inverse-free variational Bayesian method was proposed via
maximizing a relaxed evidence lower bound. Specifically, we
develop two unrolled deep networks, one for i.i.d. Gaussian
channels and the other for arbitrarily correlated channels. Our
proposed networks, referred to as the variational Bayesian
inference-inspired network (VBINet), have a very few learnable
parameters and thus can be efficiently trained with a moderate
number of training samples. The proposed VBINet can work in both
offline and online training modes. An important advantage of our
proposed networks over OAMPNet and MMNet is that the VBINet can
automatically learn the noise variance from data, which is highly
desirable in practical applications. In addition, the proposed
VBINet has a low computational complexity, and each layer of the
neural network only involves simple matrix-vector calculations.
Simulation results show that the proposed VBINet-based detector
achieves competitive performance for both i.i.d. Gaussian and 3GPP
MIMO channels. Moreover, it attains a significant performance
improvement over the OAMPNet and MMNet in the presence of noise
variance uncertainty.

The rest paper is organized as follows. Section
\ref{section-formualtion} discusses the MIMO detection problem.
Section \ref{section-overview} provides a brief overviews of some
state-of-the-art DL-based detectors. Section \ref{section-IFBL}
proposes an iterative detector for i.i.d. Gaussian channels within
the inverse-free variational Bayesian framework. A deep network
(VBINet) is then developed by unfolding the iterative detector in
Section \ref{IFBL-iid-chapter}. The iterative detector is extended
to correlated channels in Section \ref{section-arbitrary}, and an
improved VBINet is developed in Section
\ref{section-ImprovedIFBLNet} for arbitrarily correlated channels.
Numerical results are provided in Section
\ref{section-simulation}, followed by conclusion remarks in
Section \ref{section-conclusion}.

\section{Problem Formulation}   \label{section-formualtion}
We consider the problem of MIMO detection for uplink MIMO systems,
in which a base station equipped with $N_r$ antennas receives
signals from $N_t$ single-antenna users. The received signal
vector $\boldsymbol{y}\in\mathbb{C}^{N_r}$ can be expressed as
\begin{align}
\boldsymbol{y}=\boldsymbol{H}\boldsymbol{x}+\boldsymbol{n}
\label{data-model}
\end{align}
where $\boldsymbol{H}\in\mathbb{C}^{N_r\times N_t}$ denotes the
MIMO channel matrix, and
$\boldsymbol{n}\sim\mathcal{CN}(\boldsymbol{0},\frac{1}{\varepsilon}\boldsymbol{I})$
denotes the additive complex Gaussian noise. It is clear that the
likelihood function of the received signal is given by
\begin{align}
p(\boldsymbol{y}|\boldsymbol{x},\varepsilon)=
\big(\frac{\varepsilon^{\frac{1}{2}}}{\sqrt{2\pi}}\big)^{N_r}\exp\left\{-\frac{\varepsilon
\|\boldsymbol{y}-\boldsymbol{H}\boldsymbol{x}\|_2^2}{2}\right\}
\label{y-prior}
\end{align}
where each element of the detected symbol $\boldsymbol{x}$ belongs
to a discrete constellation set $\mathcal{C}=\{c_1,\ldots,c_M\}$.
In other words, the prior distribution for $\boldsymbol{x}$ is
given as
\begin{align}
p(\boldsymbol{x})=\prod_{i=1}^{N_t} p(x_i)
\end{align}
where
\begin{align}
p(x_i)=\sum_{j\in M}\frac{1}{M}\delta(x_i-c_j) \label{x-prior}
\end{align}
in which $\delta(\cdot)$ stands for the Dirac delta function.

The optimal detector for (\ref{data-model}) is the
maximum-likelihood (ML) estimator formulated as follows
\cite{ZhuMurch02, KhaniAlizadeh20}
\begin{align}
&\min_{\boldsymbol{x}\in \mathcal{C}} \quad
\|\boldsymbol{y}-\boldsymbol{H}\boldsymbol{x}\|_2^2
\end{align}
Although the ML estimator yields optimal performance, its
computational complexity is prohibitively high. To address this
difficulty, a plethora of detectors were proposed over the past
few decades to strike a balance between complexity and detection
performance. Recently, a new class of model-driven DL methods
called algorithm unrolling or unfolding were proposed for MIMO
detection. It was shown that model-driven DL-based detectors
present a significant performance improvement over traditional
iterative approaches. In the following, we first provide an
overview of state-of-the-art model-driven DL-based detectors,
namely, the OAMPNet \cite{HeWen20} and the MMNET
\cite{KhaniAlizadeh20}.

\section{Overview of State-Of-The-Art DL-Based Detectors}  \label{section-overview}
\subsection{OAMPNet}
OAMPNet is a model-driven DL algorithm for MIMO detection, which
is derived from the OAMP \cite{MaPing17}. The advantage of OAMP
over AMP is that it can be applied to unitarily-invariant matrices
while AMP is only applicable to i.i.d Gaussian measurement
matrices. OAMPNet can achieve better performance than OAMP, and
can adapt to various channel environments by use of some learnable
variables. Each iteration of the OAMPNet comprises the following
two steps:
\begin{align}
\quad & \text{LE} : \quad\quad \boldsymbol{r}_t=\boldsymbol{x}_t+
 \gamma_t\boldsymbol{W}_t(\boldsymbol{y}-\boldsymbol{H}\boldsymbol{x}_t)
\\
\quad & \text{NLE} : \quad
\boldsymbol{x}_{t+1}=\eta_t(\boldsymbol{r}_t,
 \sigma_t^2; \phi_t, \xi_t)
\end{align}
where $\boldsymbol{x}_{t}$ denotes the current estimate of
$\boldsymbol{x}$, and $\boldsymbol{W}_t$ is given as
\cite{MaPing17}:
\begin{align}
\boldsymbol{W}_t &=\frac{N_t}{\text{tr}\{\hat{\boldsymbol{W}}_t
 \boldsymbol{H}\}}\hat{\boldsymbol{W}}_t
 \label{W-t}
\end{align}
in which $\hat{\boldsymbol{W}}_t$ is the LMMSE matrix, i.e.
\begin{align}
\hat{\boldsymbol{W}}_t &=v_t^2\boldsymbol{H}^H(v_t^2
\boldsymbol{H}\boldsymbol{H}^H+\boldsymbol{R}_{n})^{-1}
\end{align}
$v_t^2$ is the error variance of
$\boldsymbol{x}_t-\boldsymbol{x}$, and can be calculated as
\begin{align}
v_t^2
&=\text{Tr}\Big\{E\big[\big(\boldsymbol{x}_t-\boldsymbol{x}\big)
 \big(\boldsymbol{x}_t-\boldsymbol{x}\big)^H\big]\Big\}
 \nonumber\\
 &=\frac{\|\boldsymbol{y}-\boldsymbol{H}\boldsymbol{x}_t\|_2^2-
\text{tr}(\boldsymbol{R}_{n})}
{\text{tr}(\boldsymbol{H}^H\boldsymbol{H})} \label{v2}
\end{align}
Here $\boldsymbol{R}_{n}$ denotes the covariance matrix of the
additive observation noise $\boldsymbol{n}$.

In the NLE step, the learnable variable $\theta_t$ is used to
regulate the error variance $\sigma_t^2$, and the error variance
$\sigma_t^2$ of $\boldsymbol{r}_t-\boldsymbol{x}$ is given by
\begin{align}
\sigma_{t}^2 =\frac{1}{N_t}\text{tr}\Big\{\boldsymbol{C}_t
\boldsymbol{C}_t\Big\}v_t^2+ \frac{\theta_t^2}{N_t}
\text{tr}\{\boldsymbol{W}_t\boldsymbol{R}_{n}\boldsymbol{W}_t^H\}
\label{sigma-variance}
\end{align}
where
$\boldsymbol{C}_t\triangleq\boldsymbol{I}-\theta_t\boldsymbol{W}_t\boldsymbol{H}
$, and (\ref{sigma-variance}) derives from the fact that
$\boldsymbol{r}_t-\boldsymbol{x}=
 \boldsymbol{C}_t
 \big(\boldsymbol{x}_t-\boldsymbol{x}\big)
 +\theta_t\boldsymbol{W}_t\boldsymbol{n}$. From (\ref{W-t}) to
(\ref{sigma-variance}), we see that calculations of
$\boldsymbol{W}_t$, $v_t^2$ and $\sigma_{t}^2$ require the
knowledge of the noise variance $1/\varepsilon$. For the NLE step,
it is a divergence-free estimator defined as
\begin{align}
\eta_t(\boldsymbol{r}_t,
 \sigma_t^2; \varphi_t, \xi_t)=\varphi_t\Big(E\{\boldsymbol{x};\boldsymbol{r}_t,
 \sigma_t^2\}-
\xi_t\boldsymbol{r}_t\Big)
\end{align}
with
\begin{align}
E\{x_i; r_{i,t},\sigma_{t}^2\} =\frac{\sum_{i}x_i N(x_i;
r_{i,t},\sigma_{t}^2) p(x_i)}{\sum_{i}N(x_i;r_{i,t},\sigma_{t}^2)
p(x_i)}
 \end{align}
where $x_i$ and $r_{i,t}$ denote the $i$th element of
$\boldsymbol{x}$ and $\boldsymbol{r}_t$, respectively, and
$N(x;\mu,\sigma^2)$ is used to represent the probability density
function of a Gaussian random variable with mean $\mu$ and
variance $\sigma^2$. For OAMPNet, each layer has four learnable
parameters, i.e. $\{\gamma_t, \theta_t, \varphi_t, \xi_t\}$, in
which $\gamma_t$ and $\theta_t$ are used to optimize the LE
estimator, whereas $\{\varphi_t, \xi_t\}$ are employed to adjust
the NLE estimator.

\subsection{MMNet}
MMNet has two variants, one for i.i.d Gaussian channels and the
other for arbitrarily correlated channels. Similar to OAMPNet, the
noise variance is also assumed known in advance for MMNet. Each
layer of the MMNet performs the following update:
\begin{align}
\quad & \text{LE} : \quad\quad
\boldsymbol{r}_t=\boldsymbol{x}_t+\boldsymbol{A}_t\left(
 \boldsymbol{y}-\boldsymbol{H}\boldsymbol{x}_t\right)
\\
\quad & \text{NLE} : \quad
\boldsymbol{x}_{t+1}=E\{\boldsymbol{x};\boldsymbol{r}_t,
 \boldsymbol{\sigma}_t^2\}
\end{align}
Here $\boldsymbol{A}_t$ is chosen to be
$\boldsymbol{A}_t=\theta_t^{(1)}\boldsymbol{H}^{H}$ for i.i.d
Gaussian channels and an entire trainable matrix for arbitrary
channels. The error variance vector $\boldsymbol{\sigma}_t^2$ can
be calculated as
\begin{align}
\boldsymbol{\sigma}_t^2=\frac{\boldsymbol{\theta}_t^{(2)}}{N_t}\Big(
\frac{\|\boldsymbol{I}-\boldsymbol{A}_t\boldsymbol{H}\|_F^2}
{\|\boldsymbol{H}\|_F^2}
\big[\|\boldsymbol{y}-\boldsymbol{H}\boldsymbol{x}_t\|_2^2-
\frac{N_r}{\varepsilon}\big]_{+}+\frac{\|\boldsymbol{A}_t\|_F^2}{\varepsilon}\Big)
\end{align}
where $[x]_{+}=\max(x,0)$, and MMNet chooses
$\boldsymbol{\theta}_t^{(2)}$ as
$\boldsymbol{\theta}_t^{(2)}\triangleq\theta_t^{(2)}\cdot
\boldsymbol{1}_{N_t}$ for i.i.d. Gaussian channels and an entire
trainable vector for arbitrary channels. Here, $\theta_t^{(2)}$ is
a learnable parameter, and $\boldsymbol{1}_{N_t}$ is an all-one
vector.

A major drawback of MMNet is that, for the correlated channel
case, MMNet needs to train model parameters for each realization
of $\boldsymbol{H}$. The reason is that for a specific
$\boldsymbol{A}_t$, it is difficult to have the neural network
generalized to a variety of channel realizations. Therefore, MMNet
achieves impressive performance when online training is allowed,
where the model parameters can be trained and tested for each
realization of the channel; while MMNet fails to work when only
offline training is allowed. Here offline training means that
training is performed over randomly generated channels, and its
performance is evaluated over another set of randomly generated
channels.

\section{IFVB Detector} \label{section-IFBL}
Our proposed DL architecture is mainly inspired by the
inverse-free variational Bayesian learning (IFVB) framework
developed in \cite{DuanYang17}, where a fast inverse-free
variational Bayesian method was proposed via maximizing a relaxed
evidence lower bound. In this section, we first develop a Bayesian
MIMO detector within the IFVB framework.

\subsection{Overview of Variational Inference} We first provide a
review of the variational Bayesian inference. Let
$\boldsymbol{\theta}\triangleq\{\boldsymbol{\theta}_1,\ldots,\boldsymbol{\theta}_I\}$
denote the hidden variables in the model. The objective of
Bayesian inference is to find the posterior distribution of the
latent variables given the observed data, i.e.
$p(\boldsymbol{\theta}|\boldsymbol{y})$. The computation of
$p(\boldsymbol{\theta}|\boldsymbol{y})$, however, is usually
intractable. To address this difficulty, in variational inference,
the posterior distribution $p(\boldsymbol{\theta}|\boldsymbol{y})$
is approximated by a variational distribution
$q(\boldsymbol{\theta})$ that has a factorized form as
\cite{TzikasLikas08}
\begin{align}
q(\boldsymbol{\theta})=\prod_{i=1}^I q_i(\boldsymbol{\theta}_i)
\end{align}
Note that the marginal probability of the observed data can be
decomposed into two terms
\begin{align}
\ln p(\boldsymbol{y})=L(q)+\text{KL}(q|| p)
\label{variational-decomposition}
\end{align}
where
\begin{align}
L(q)=\int q(\boldsymbol{\theta})\ln
\frac{p(\boldsymbol{y},\boldsymbol{\theta})}{q(\boldsymbol{\theta})}d\boldsymbol{\theta}
\end{align}
and
\begin{align}
\text{KL}(q|| p)=-\int q(\boldsymbol{\theta})\ln
\frac{p(\boldsymbol{\theta}|\boldsymbol{y})}{q(\boldsymbol{\theta})}d\boldsymbol{\theta}
\end{align}
where $\text{KL}(q|| p)$ is the Kullback-Leibler divergence
between $p(\boldsymbol{\theta}|\boldsymbol{y})$ and
$q(\boldsymbol{\theta})$, and $L(q)$ is the evidence lower bound
(ELBO). Since $\ln p(\boldsymbol{y})$ is a constant, maximizing
$L(q)$ is equivalent to minimizing $\text{KL}(q|| p)$,
 and thus
the posterior distribution $p(\boldsymbol{\theta}|\boldsymbol{y})$
can be approximated by the variational distribution
$q(\boldsymbol{\theta})$ through maximizing $L(q)$. The ELBO
maximization can be conducted in an alternating fashion for each
latent variable, which leads to \cite{TzikasLikas08}
\begin{align}
q_i(\boldsymbol{\theta}_i)=\frac{\exp(\langle\ln
p(\boldsymbol{y},\boldsymbol{\theta})\rangle_{k\neq
i})}{\int\exp(\langle\ln
p(\boldsymbol{y},\boldsymbol{\theta})\rangle_{k\neq
i})d\boldsymbol{\theta}_i} \label{general-update}
\end{align}
where $\langle\cdot\rangle_{k\neq i}$ denotes an expectation with
respect to the distributions $q_k(\boldsymbol{\theta}_k)$ for all
$k\neq i$.

\subsection{IFVB-Based MIMO Detector}
To facilitate the Bayesian inference, we place a $\text{Gamma}$
hyperprior over the hyperparameter $\varepsilon$, i.e.
\begin{align}
p(\varepsilon)=\text{Gamma}(\varepsilon|a,b) \label{prior-epsilon}
\end{align}
where $a$ and $b$ are small positive constants, e.g. $a=b=10^{-10}$. We have
\begin{align}
L(q)&=\int
q(\boldsymbol{\theta})\ln\frac{p(\boldsymbol{y},\boldsymbol{\theta})}
{q(\boldsymbol{\theta})}d\boldsymbol{\theta}
\nonumber\\
&=\int
q(\boldsymbol{\theta})\ln\frac{p(\boldsymbol{y}|\boldsymbol{x},\varepsilon)p(\boldsymbol{x})
p(\varepsilon)} {q(\boldsymbol{\theta})}d\boldsymbol{\theta}
\end{align}
Let $\boldsymbol{\theta}=\{\boldsymbol{x},\varepsilon\}$ denote
the hidden variables in the MIMO detection problem, and let the
variational distribution be expressed as
$q(\boldsymbol{\theta})=q_{x}(\boldsymbol{x})q_{\varepsilon}(\varepsilon)$.
Variational inference is conducted by maximizing the ELBO, which
yields a procedure which involves updates of the approximate
posterior distributions for hidden variables $\boldsymbol{x}$ and
$\varepsilon$ in an alternating fashion. Nevertheless, it is
difficult to derive the approximate posterior distribution for
$\boldsymbol{x}$ due to the coupling of elements of
$\boldsymbol{x}$ in the likelihood function
$p(\boldsymbol{y}|\boldsymbol{x},\varepsilon)$.

To address this difficulty, we, instead, maximize a lower bound on
the ELBO $L(q)$, also referred to as a relaxed ELBO. To obtain a
relaxed ELBO, we first recall the following fundamental property
for a smooth function.

\newtheorem{lemma}{Lemma}
\begin{lemma} \label{lemma1}
If the continuously differentiable function
$f:\mathbb{R}^n\rightarrow \mathbb{R}$ has bounded curvature, i.e.
there exists a matrix $\boldsymbol{T}$ such that
$\boldsymbol{T}\succcurlyeq \frac{\nabla^2
f(\boldsymbol{x})}{2},\forall \boldsymbol{x}$, then, for any
$\boldsymbol{u}, \boldsymbol{v}\in\mathbb{R}^n$, we have
\begin{align}
f(\boldsymbol{u})\leqslant
f(\boldsymbol{v})+(\boldsymbol{u}-\boldsymbol{v})^H\nabla
f(\boldsymbol{v})
+(\boldsymbol{u}-\boldsymbol{v})^H\boldsymbol{T}(\boldsymbol{u}-\boldsymbol{v})
\end{align}
\end{lemma}
\begin{proof}
See \cite{SunBabu16, MagnusNeudecker19}.
\end{proof}

Invoking Lemma \ref{lemma1}, a lower bound on
$p(\boldsymbol{y}|\boldsymbol{x},\varepsilon)$ can be obtained as
\begin{align}
p(\boldsymbol{y}|\boldsymbol{x},\varepsilon)&=
\big(\frac{\varepsilon^{\frac{1}{2}}}{\sqrt{2\pi}}\big)^{N_r}
\exp\Big\{-\frac{\varepsilon}{2}\|\boldsymbol{y}-\boldsymbol{H}\boldsymbol{x}\|_2^2\Big\}
\nonumber\\
&\geqslant\big(\frac{\varepsilon^{\frac{1}{2}}}{\sqrt{2\pi}}\big)^{N_r}
\exp\Big\{-\frac{\varepsilon}{2}g(\boldsymbol{x},\boldsymbol{z})\Big\}
\triangleq
F(\boldsymbol{y},\boldsymbol{x},\varepsilon,\boldsymbol{z})
\label{g-y-inequality}
\end{align}
where
\begin{align}
g(\boldsymbol{x},\boldsymbol{z}) &
=\|\boldsymbol{y}-\boldsymbol{H}\boldsymbol{z}\|_2^2+
2\Re\big\{(\boldsymbol{x}-\boldsymbol{z})^H\boldsymbol{H}^H(\boldsymbol{H}\boldsymbol{z}-\boldsymbol{y})
\big\}
 \nonumber\\
&
 \quad +(\boldsymbol{x}-\boldsymbol{z})^H\boldsymbol{T}
 (\boldsymbol{x}-\boldsymbol{z})
   \end{align}
Here $\boldsymbol{T}$ needs to satisfy $\boldsymbol{T}\succcurlyeq
\boldsymbol{H}^H\boldsymbol{H}$, and the inequality
(\ref{g-y-inequality}) becomes equality when $\boldsymbol{T}=
\boldsymbol{H}^H\boldsymbol{H}$. Finding a matrix $\boldsymbol{T}$
satisfying $\boldsymbol{T}\succcurlyeq
\boldsymbol{H}^H\boldsymbol{H}$ is not difficult. In particular,
we wish $\boldsymbol{T}$ to be a diagonal matrix in order to help
obviate the difficulty of matrix inversion. An appropriate choice
of $\boldsymbol{T}$ is
\begin{align}
\boldsymbol{T}=(\lambda_{\text{max}}(\boldsymbol{H}^H\boldsymbol{H})+\epsilon)\boldsymbol{I}
\label{T-choice}
\end{align}
where $\lambda_{\text{max}}(\boldsymbol{A})$ denotes the largest
eigenvalue of $\boldsymbol{A}$, and $\epsilon$ is a small positive
constant, say, $\epsilon=10^{-10}$.

Utilizing (\ref{g-y-inequality}), a relaxed ELBO on $L(q)$ can be
obtained as
\begin{align}
L(q)\geqslant \tilde{L}(q,\boldsymbol{z})&=\int
q(\boldsymbol{\theta})\ln\frac{G(\boldsymbol{y},\boldsymbol{\theta},\boldsymbol{z})}
{q(\boldsymbol{\theta})}d\boldsymbol{\theta}
\end{align}
where
\begin{align}
G(\boldsymbol{y},\boldsymbol{\theta},\boldsymbol{z})\triangleq
F(\boldsymbol{y},\boldsymbol{x},\varepsilon,\boldsymbol{z})
p(\boldsymbol{x})p(\varepsilon) \label{G-prior}
\end{align}
In order to satisfy a rigorous distribution, the relaxed ELBO can
be further expressed as
\begin{align}
 \tilde{L}(q,\boldsymbol{z})&=\int q(\boldsymbol{\theta})\ln
 \frac{G(\boldsymbol{y},\boldsymbol{\theta},\boldsymbol{z})}
{q(\boldsymbol{\theta})}d\boldsymbol{\theta}
\nonumber\\
&=\int
q(\boldsymbol{\theta})\ln\frac{G(\boldsymbol{y},\boldsymbol{\theta},\boldsymbol{z})h(\boldsymbol{z})}
{q(\boldsymbol{\theta})h(\boldsymbol{z})}d\boldsymbol{\theta}
\nonumber\\
&=\int
q(\boldsymbol{\theta})\ln\frac{\tilde{G}(\boldsymbol{y},\boldsymbol{\theta},\boldsymbol{z})}
{q(\boldsymbol{\theta})}d\boldsymbol{\theta}-\ln h(\boldsymbol{z})
\label{L-cost}
\end{align}
where
\begin{align}
\tilde{G}(\boldsymbol{y},\boldsymbol{\theta},\boldsymbol{z})
\triangleq
G(\boldsymbol{y},\boldsymbol{\theta},\boldsymbol{z})h(\boldsymbol{z})
\end{align}
and the normalization term $h(\boldsymbol{z})$ guarantees that
$\tilde{G}(\boldsymbol{y},\boldsymbol{\theta},\boldsymbol{z})$
follows a rigorous distribution, and the exact expression of
$h(\boldsymbol{z})$ is
\begin{align}
h(\boldsymbol{z})\triangleq \frac{1}{\int
G(\boldsymbol{y},\boldsymbol{\theta},\boldsymbol{z})
d\boldsymbol{\theta}d\boldsymbol{y}}
\end{align}

Our objective is to maximize the relaxed ELBO
$\tilde{L}(q,\boldsymbol{z})$ with respect to
$q_{\theta}(\boldsymbol{\theta})$ and the parameter
$\boldsymbol{z}$. This naturally leads to a variational
expectation-maximization (EM) algorithm. In the E-step, the
posterior distribution approximations are computed in an
alternating fashion for each hidden variable, with other variables
fixed. In the M-step, $\tilde{L}(q,\boldsymbol{z})$ is maximized
with respect to $\boldsymbol{z}$, given $q(\boldsymbol{\theta})$
fixed. Details of the Bayesian inference are provided below.

\textbf{A. E-step}

1) \emph{Update of $q_x{(\boldsymbol{x})}$}: Ignoring those terms
that are independent of $\boldsymbol{x}$, the approximate
posterior distribution $q_x{(\boldsymbol{x})}$ can be obtained as
\begin{align}
\ln q_{\boldsymbol{x}}(\boldsymbol{x})
 &\propto
 \big\langle\ln\tilde{G}(\boldsymbol{y},\boldsymbol{\theta},\boldsymbol{z})
 \big\rangle_{q_{\varepsilon}(\varepsilon)}
 \nonumber\\
 &\propto
 \big\langle
 \ln F(\boldsymbol{y},\boldsymbol{x},\varepsilon,\boldsymbol{z})
 +\ln
p(\boldsymbol{x})\big\rangle_{q_{\varepsilon}(\varepsilon)}
 \nonumber\\
 &\propto
 \big\langle
 -\frac{\varepsilon}{2} g(\boldsymbol{x},\boldsymbol{z}) +
 \ln
p(\boldsymbol{x})\big\rangle_{q_{\varepsilon}(\varepsilon)}
\nonumber\\
 &\propto
 \Big\langle
 -\frac{\varepsilon}{2}\Big(
2\Re\big\{(\boldsymbol{x}-\boldsymbol{z})^H\boldsymbol{H}^H(\boldsymbol{H}\boldsymbol{z}-
\boldsymbol{y})\big\} +
\nonumber\\
&\quad\quad\quad (\boldsymbol{x}-\boldsymbol{z})^H\boldsymbol{T}
 (\boldsymbol{x}-\boldsymbol{z}) \Big) +\ln
p(\boldsymbol{x})\Big\rangle_{q_{\varepsilon}(\varepsilon)}
\nonumber\\
 &\propto
 \ln N(\boldsymbol{x};\boldsymbol{r},\boldsymbol{\Phi})
 +\ln
p(\boldsymbol{x})
\end{align}
where
\begin{align}
 \boldsymbol{\Phi} &=\frac{1}{\langle\varepsilon\rangle}\boldsymbol{T}^{-1} \label{eqn3}\\
 \boldsymbol{r} &=\langle\varepsilon\rangle\boldsymbol{\Phi}\Big(
 \boldsymbol{H}^{H}\boldsymbol{y}+\boldsymbol{T}
 \boldsymbol{z}-\boldsymbol{H}^{H}
 \boldsymbol{H}\boldsymbol{z}\Big)
 \nonumber\\
 &=\boldsymbol{z}+\boldsymbol{T}^{-1}\boldsymbol{H}^{H}\Big(
 \boldsymbol{y}-\boldsymbol{H}\boldsymbol{z}\Big)
 \label{x-update1}
 \end{align}
Since $\boldsymbol{T}$ is chosen to be a diagonal matrix,
calculation of $\boldsymbol{\Phi}$ is simple and no longer
involves computing the inverse of a matrix. Also, as the
covariance matrix $\boldsymbol{\Phi}$ is a diagonal matrix and the
prior distribution $p(\boldsymbol{x})$ has a factorized form, the
approximate posterior $q_x{(\boldsymbol{x})}$ also has a
factorized form, which implies that elements of $\boldsymbol{x}$
are mutually independent.

Thus, the mean of the $i$th element of $\boldsymbol{x}$ can be
easily calculated as
 \begin{align}
 \langle x_i\rangle=E\{x_i;r_i,\Phi_i\}=
 \frac{\sum_{i}x_i N(x_i; r_{i},\Phi_{i}) p(x_i)}{\sum_{i}N(x_i;r_{i},\Phi_{i}) p(x_i)},\quad \forall i
 \label{x-est1}
 \end{align}
 where $r_{i}$ is the $i$th element of $\boldsymbol{r}$, $\Phi_{i}$
is the $i$th diagonal element of  $\boldsymbol{\Phi}$. The
expectation of $x_i^2$ can be computed as
 \begin{align}
 \langle x_i^2\rangle=E\{x_i^2;r_i,\Phi_i\}=
 \frac{\sum_{i}x_i^2 N(x_i; r_{i},\Phi_{i}) p(x_i)}
 {\sum_{i}N(x_i;r_{i},\Phi_{i}) p(x_i)} ,\quad \forall i
  \label{x2-est1}
 \end{align}

2) \emph{Update of $q_{\varepsilon}{(\varepsilon)}$}: The
variational distribution $q_{\varepsilon}{(\varepsilon)}$ can be
calculated by
\begin{align}
\ln q_\varepsilon(\varepsilon)&\propto \langle
\ln\tilde{G}(\boldsymbol{y},\boldsymbol{\theta},\boldsymbol{z})
\rangle_{q_{x}(\boldsymbol{x})}
\nonumber\\
 &\propto
 \big\langle
 \ln F(\boldsymbol{y},\boldsymbol{x},\varepsilon,\boldsymbol{z})+\ln
p(\varepsilon)\big\rangle_{q_{x}(\boldsymbol{x})}
\nonumber\\
 &\propto
 \Big(a-1+\frac{N_r}{2}\Big)\ln \varepsilon
 -\Big(\frac{1}{2}\big\langle g(\boldsymbol{x},\boldsymbol{z})\big\rangle+b\Big)\varepsilon
 \end{align}
Thus $\varepsilon$ follows a Gamma distribution given as
\begin{align}
q_{\varepsilon}(\varepsilon)=\text{Gamma}(\varepsilon;\tilde{a},\tilde{b})
\end{align}
in which
\begin{align}
\tilde{a}&=a+\frac{N_r}{2}
\nonumber\\
\tilde{b}&=b+\frac{1}{2}\langle
g(\boldsymbol{x},\boldsymbol{z})\rangle \label{iid-epsilon1}
\end{align}
and
\begin{align}
\langle g(\boldsymbol{x},\boldsymbol{z})\rangle
 &=
 \|\boldsymbol{y}-\boldsymbol{H}\boldsymbol{z}\|_2^2+
2\Re\Big\{(\langle\boldsymbol{x}\rangle-\boldsymbol{z})^H\boldsymbol{H}^H
(\boldsymbol{H}\boldsymbol{z}-\boldsymbol{y})\Big\}
\nonumber\\
& \quad +\text{Tr}\Big(\boldsymbol{T}
\big(\langle\boldsymbol{x}\rangle-\boldsymbol{z}\big)
\big(\langle\boldsymbol{x}\rangle-\boldsymbol{z}\big)^H
 +\boldsymbol{T}\boldsymbol{\Sigma}_x
 \Big)
\end{align}
where $\boldsymbol{\Sigma}_x$ denotes the covariance matrix of
$\boldsymbol{x}$. Due to the independence among entries in
$\boldsymbol{x}$, $\boldsymbol{\Sigma}_x$ is a diagonal matrix and
its $i$th diagonal element is given by
\begin{align}
\Sigma_{x_i}=\langle x_i^2\rangle-\langle x_i\rangle^2,\quad \forall i
\end{align}

Finally, we have
\begin{align}
\langle\varepsilon\rangle=\frac{\tilde{a}}{\tilde{b}}
\label{iid-epsilon2}
\end{align}

\textbf{B. M-step}

Substituting $q(\boldsymbol{\theta}; \boldsymbol{z}^{\text{old}})$
into $\tilde{L}(q,\boldsymbol{z})$, an estimate of
$\boldsymbol{z}$ can be found via the following optimization
\begin{align}
\boldsymbol{z}^{\text{new}}&=\arg \max_{\boldsymbol{z}}\langle
\ln{G}(\boldsymbol{y},\boldsymbol{\theta},\boldsymbol{z})
\rangle_{q(\boldsymbol{\theta};\boldsymbol{z}^{\text{old}})}\triangleq
Q(\boldsymbol{z}|\boldsymbol{z}^{\text{old}})
\end{align}
Setting the derivative of the logarithm function to zero yields:
\begin{align}
\frac{\partial
\langle\ln{G}(\boldsymbol{y},\boldsymbol{\theta},\boldsymbol{z})
\rangle_{q(\boldsymbol{\theta};\boldsymbol{z}^{\text{old}})}}{\partial
\boldsymbol{z}}
=\langle\varepsilon\rangle\Big(\boldsymbol{T}-\boldsymbol{H}^H\boldsymbol{H}\Big)
(\langle\boldsymbol{x}\rangle-\boldsymbol{z}) =\boldsymbol{0}
\end{align}
As $\langle\varepsilon\rangle>0$ and $\boldsymbol{T}\succeq
\boldsymbol{H}^H\boldsymbol{H}$, the solution of $\boldsymbol{z}$
is
\begin{align}
\boldsymbol{z}^{\text{new}}=\langle\boldsymbol{x}\rangle
\label{iid-z-update}
\end{align}

\subsection{Summary of IFVB Detector}   \label{sum-IFBL}
For clarity, the proposed IFVB detector is summarized as an
iterative algorithm with each iteration consisting of two steps,
namely, a linear estimation (LE) step and a nonlinear estimation
(NLE) step:
\begin{align}
\quad & \text{LE} : \quad\quad \boldsymbol{r}_t=\boldsymbol{x}_t+
\boldsymbol{T}^{-1}\boldsymbol{H}^{H}(
 \boldsymbol{y}-\boldsymbol{H}\boldsymbol{x}_t) \label{LE-IFBL}
\\
\quad & \text{NLE} : \quad
\boldsymbol{x}_{t+1}=E\{\boldsymbol{x};\boldsymbol{r}_t,
\boldsymbol{\Phi}_t\},
 \label{NLE-IFBL}
\end{align}
where $\boldsymbol{r}_t$ and $\boldsymbol{\Phi}_t$ denote the
current estimate of $\boldsymbol{r}$ and $\boldsymbol{\Phi}$,
respectively. The linear estimator in equation (\ref{LE-IFBL})
derives from (\ref{x-update1}) and (\ref{iid-z-update}). Here with
a slight abuse of notation, we use $\boldsymbol{x}_t$ to represent
the current estimate of $\langle\boldsymbol{x}\rangle$. The
nonlinear estimator (\ref{NLE-IFBL}) is a vector form of
(\ref{x-est1}). In (\ref{NLE-IFBL}), the covariance matrix
$\boldsymbol{\Phi}_t$ can be calculated as $\boldsymbol{\Phi}_t
=\frac{1}{\varepsilon_t}\boldsymbol{T}^{-1}$ (cf. (\ref{eqn3})),
with $\varepsilon_{t+1}$ updated as
\begin{align}
\varepsilon_{t+1}=\frac{\tilde{a}}{\tilde{b}_{t+1}}=\frac{a+\frac{N_r}{2}}{b+\frac{1}{2}
g(\boldsymbol{x}_{t+1},\boldsymbol{x}_t)}
\end{align}
where
\begin{align}
& g(\boldsymbol{x}_{t+1},\boldsymbol{x}_t)
\nonumber\\
 &=
 \|\boldsymbol{y}-\boldsymbol{H}\boldsymbol{x}_t\|_2^2+
2\Re\Big\{(\boldsymbol{x}_{t+1}-\boldsymbol{x}_t)^H\boldsymbol{H}^H
(\boldsymbol{H}\boldsymbol{x}_t-\boldsymbol{y})\Big\}
\nonumber\\
& \quad +\text{Tr}\Big(\boldsymbol{T}
\big(\boldsymbol{x}_{t+1}-\boldsymbol{x}_t\big)
\big(\boldsymbol{x}_{t+1}-\boldsymbol{x}_t\big)^H
 +\boldsymbol{T}\boldsymbol{\Sigma}_x^{t+1}
 \Big)
\end{align}
and the $i$th diagonal element of $\boldsymbol{\Sigma}_x^{t+1}$
can be calculated as
\begin{align}
\Sigma_{x_i}^{t+1}=
E\{x_{i}^2;r_{i,t},\Phi_{i,t}\}-\big[E\{x_{i};r_{i,t},\Phi_{i,t}\}\big]^2,\quad
\forall i \label{eqn1}
\end{align}
in which $r_{i,t}$ and $\Phi_{i,t}$ denote the $i$th elements of
$\boldsymbol{r}_t$ and $\boldsymbol{\Phi}_t$, respectively.

\begin{figure*}[t]
\setlength{\abovecaptionskip}{0pt}
\setlength{\belowcaptionskip}{0pt} \centering
\includegraphics[width=12cm]{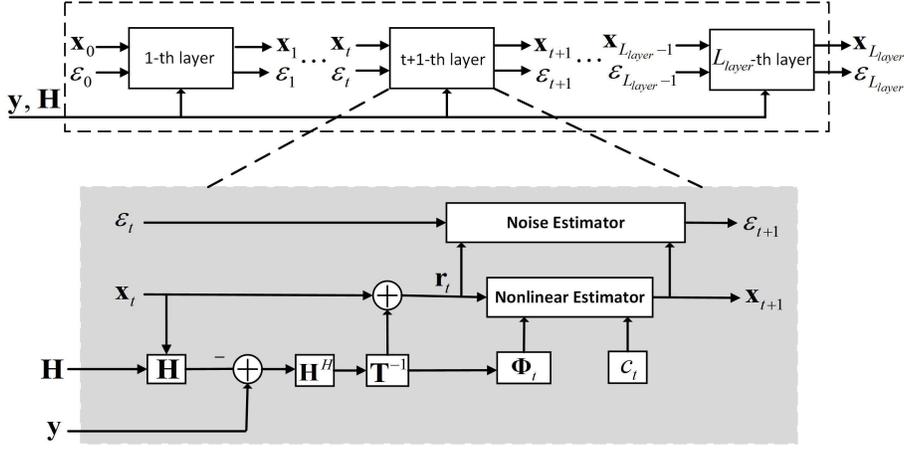}
\caption{Block diagram of the proposed VBINet detector.}
\label{IID-network}
\end{figure*}

\section{Proposed VBINet-Based Detector}   \label{IFBL-iid-chapter}
In this section, we propose a model-driven DL detector (referred
to as VBINet) which is developed by unfolding the iterative IFVB
detector. The network architecture is illustrated in Fig.
\ref{IID-network}, in which the network consists of
$L_{\text{layer}}$ cascade layers with some trainable parameters.
The trainable parameters of all layers are denoted as
$\boldsymbol{\Omega}\triangleq\big\{\boldsymbol{T},
\{c_t\}_{t=1}^{L_\text{layer}}\big\}$, where $\boldsymbol{T}$ is a
diagonal matrix common to all layers. We will explain these
trainable parameters in detail later. For the $(t+1)$th layer, the
input includes $\boldsymbol{y}$, $\boldsymbol{H}$,
$\boldsymbol{x}_{t}$, and $\varepsilon_{t}$, where
$\boldsymbol{x}_{t}$ and $\varepsilon_{t}$ denote the $t$th
layer's estimate of the signal and the noise variance,
respectively. Given
$\{\boldsymbol{y},\boldsymbol{H},\boldsymbol{x}_{t},\varepsilon_{t}\}$,
each layer performs the following updates:
\begin{align}
\quad & \text{LE} : \quad\quad \boldsymbol{r}_t=\boldsymbol{x}_t+
\boldsymbol{T}^{-1}\boldsymbol{H}^{H}\left(
 \boldsymbol{y}-\boldsymbol{H}\boldsymbol{x}_t\right) \label{LE-IFBLNet},
\\
\quad & \text{NLE} : \quad
 \boldsymbol{x}_{t+1}=c_t E\{\boldsymbol{x};\boldsymbol{r}_t,  \boldsymbol{\Phi}_t\}
 +(1-c_t) \boldsymbol{x}_t,
 \label{NLE-IFBLNet} \\
\quad & \text{Update of $\varepsilon_{t+1}$} :\quad
\varepsilon_{t+1}=\frac{a+\frac{N_r}{2}}{b+\frac{1}{2}
g_{\text{net}}(\boldsymbol{x}_{t+1},\boldsymbol{x}_t)}
\end{align}
where
\begin{align}
&\quad\quad\quad\quad\quad\quad\quad\quad\boldsymbol{\Phi}_t
=\frac{1}{\varepsilon_t}\boldsymbol{T}^{-1}
,\\
& g_{\text{net}}(\boldsymbol{x}_{t+1},\boldsymbol{x}_t)
\nonumber\\
& =
 \|\boldsymbol{y}-\boldsymbol{H}\boldsymbol{x}_t\|_2^2+
2\Re\Big\{(\boldsymbol{x}_{t+1}-\boldsymbol{x}_t)^H\boldsymbol{H}^H
(\boldsymbol{H}\boldsymbol{x}_t-\boldsymbol{y})\Big\}
\nonumber\\
& \quad +\text{Tr}\Big(\boldsymbol{T}
\big(\boldsymbol{x}_{t+1}-\boldsymbol{x}_t\big)
\big(\boldsymbol{x}_{t+1}-\boldsymbol{x}_t\big)^H
 +\boldsymbol{T}\boldsymbol{\Sigma}_{\text{net}}^{t+1}
 \Big),
 \\
 &\Sigma_{i, {\text{net}}}^{t+1}=   c_t^2
\Big(E\{x_{i}^2;r_{i,t},\Phi_{i,t}\}-\big[E\{x_{i};r_{i,t},\Phi_{i,t}\}\big]^2\Big)
,\quad \forall i.\label{Sigma-net-i}
\end{align}
in which $\Sigma_{i, {\text{net}}}^{t+1}$ is the $i$th diagonal
element of the diagonal matrix
$\boldsymbol{\Sigma}_{\text{net}}^{t+1}$.

We see that the update formulas
(\ref{LE-IFBLNet})--(\ref{Sigma-net-i}) are similar to the IFVB
detector's update formulas (\ref{LE-IFBL})--(\ref{eqn1}).
Nevertheless, there are two major differences. Firstly, in the LE
step, the diagonal matrix $\boldsymbol{T}$ is no longer
pre-specified; instead, it is a trainable diagonal matrix. The
reason for doing this is to make the network learnable and
flexible. Specifically, we believe that $\boldsymbol{T}$ given in
(\ref{T-choice}) may not be the best, and hope that we can find a
more suitable diagonal matrix through training. Secondly, in the
NLE step, $\boldsymbol{x}_{t+1}$ is updated by applying an
adaptive damping scheme with a learnable parameter $c_t$. In the
adaptive damping scheme, the damping parameter can help control
the convergence rate and improve the robustness of the system
\cite{RanganSchniter19}.

One may wonder why not use an individual $\boldsymbol{T}$ for each
layer of the network. The reason is that using a common diagonal
matrix $\boldsymbol{T}$ for all layers can substantially reduce
the number of learnable variables, which facilitates training when
the number of training data is limited. Also, notice that the
iterative algorithm developed in the previous section uses a same
$\boldsymbol{T}$ throughout its iterative process. This suggests
that using an individual $\boldsymbol{T}$ for each layer may not
be necessary. To corroborate our conjecture, we conducted
experiments to compare these two schemes. Simulation results show
that using an individual $\boldsymbol{T}$ for each layer does not
achieve a clear performance improvement over the scheme of using a
common diagonal matrix $\boldsymbol{T}$.

Moreover, due to the use of the term $(1-c_t)\boldsymbol{x}_t$ in
the NLE step, the covariance matrix of $\boldsymbol{x}_{t+1}$,
denoted by $\boldsymbol{\Sigma}_x$, cannot be updated via
(\ref{eqn1}). Note that in the IFVB detector, according to
(\ref{NLE-IFBL}), the NLE outputs the mean value of the variable
$\boldsymbol{x}$. Meanwhile, in the VBINet, the output of the NLE
is the mean of a new variable $\boldsymbol{x}_{\text{net}}$
defined as
\begin{align}
\boldsymbol{x}_{\text{net}}= c_t \boldsymbol{x}+ (1-c_t)
\boldsymbol{x}_t
\end{align}
Thus the mean of $\boldsymbol{x}_{\text{net}}$ is given by
\begin{align}
\boldsymbol{x}_{t+1}=\langle\boldsymbol{x}_{\text{net}}\rangle=
c_t E\{\boldsymbol{x};\boldsymbol{r}_t, \boldsymbol{\Phi}_t\}+
(1-c_t) \boldsymbol{x}_t \label{iid-xnew}
\end{align}
and the $i$th diagonal element of the diagonal covariance matrix
$\boldsymbol{\Sigma}_{\text{net}}^{t+1}$ of the random variable
$\boldsymbol{x}_{\text{net}}$ is given by
\begin{align}
\Sigma_{i, {\text{net}}}^{t+1} &=\text{var}(x_{i,\text{net}})
\nonumber\\
&=  c_t^2  (\langle x_{i}^2\rangle-\langle x_i\rangle^2)
\nonumber\\
&= c_t^2
\Big(E\{x_{i}^2;r_{i,t},\Phi_{i,t}\}-\big[E\{x_{i};r_{i,t},\Phi_{i,t}\}\big]^2\Big)
,\quad \forall i
\end{align}
in which $x_{i,\text{net}}$ denotes the $i$th element of
$\boldsymbol{x}_{\text{net}}$. This explains how
(\ref{Sigma-net-i}) is derived.

\emph{Remark 1:} We see that the total number of learnable
parameters of VBINet is $|\boldsymbol{\Omega}|= N_t +
L_{\text{layer}}$, as each layer shares the same learnable
variable $\boldsymbol{T}$ while $c_t$ are generally different for
different layers. Suppose that the number of users is $N_t=16$ and
the number of layers is set to $L_{\text{layer}}=10$. The total
number of learnable parameters of VBINet is $26$. As a comparison,
the OAMPNet has $40$ parameters to be learned for the same problem
being considered. Due to the small number of learnable parameters,
the proposed VBINet can be effectively trained even with a limited
number of training data samples.

\emph{Remark 2:} An important advantage of the proposed VBINet
over state-of-the-art deep unfolding-based detectors such as the
OAMPNet and the MMNet is that the VBINet can automatically
estimate the noise variance $1/\varepsilon$, while both the
OAMPNet and the MMNet require the knowledge of $1/\varepsilon$,
and an inaccurate knowledge of the noise variance results in a
considerable amount of performance degradation, as will be
demonstrated later in our experiments. Another advantage of the
proposed VBINet is that the update formulas can be efficiently
implemented via simple matrix/vector operations. No matrix inverse
operation is needed.

\emph{Remark 3:} The proposed VBINet can provide an approximate
posterior probability of the transmitted symbols, which is useful
for channel decoders relying on the estimate of the posterior
probability of the transmitted symbols \cite{ChenWang00}. Thus the
proposed VBINet can be well suited for joint symbol detection and
channel decoding.

\section{IFVB Detector: Extension To Correlated Channels} \label{section-arbitrary}
The proposed IFVB detector performs well on channel matrices with
independent and identically distributed (i.i.d.) entries.
Nevertheless, similar to AMP and other message-passing-based
algorithms, its performance degrades considerably on realistic,
spatially-correlated channel matrices. To address this issue, we
discuss how to improve the proposed IFVB detector to accommodate
arbitrarily correlated channel matrices.

The key idea behind the improved IFVB detector is to conduct a
truncated singular value decomposition (SVD) of the channel matrix
$\boldsymbol{H}=\boldsymbol{U}\boldsymbol{\Sigma}\boldsymbol{V}^H$,
where $\boldsymbol{U}\in\mathbb{C}^{N_r\times N_t}$,
$\boldsymbol{\Sigma}\in\mathbb{C}^{N_t\times N_t}$, and
$\boldsymbol{V}\in\mathbb{C}^{N_t\times N_t}$. Define
$\boldsymbol{A}\triangleq \boldsymbol{U}\boldsymbol{\Sigma}$ and
$\boldsymbol{s}\triangleq \boldsymbol{V}^H\boldsymbol{x}$. The
received signal vector $\boldsymbol{y}\in\mathbb{C}^{N_r}$ in
(\ref{data-model}) can be equivalently written as
 \begin{align}
\boldsymbol{y}=\boldsymbol{A}\boldsymbol{s}+\boldsymbol{n}
\label{data-model-new}
\end{align}
Instead of detecting $\boldsymbol{x}$, we aim to estimate
$\boldsymbol{s}$ based on the above observation model. Let
$\boldsymbol{\theta}\triangleq\{\boldsymbol{s},\varepsilon\}$
denote the hidden variables, and let $p(\boldsymbol{s})$ denote
the prior distribution of $\boldsymbol{s}$. Thus the ELBO is given
as
\begin{align}
L(q)&\triangleq \int q(\boldsymbol{\theta})
\ln\frac{p(\boldsymbol{y},\boldsymbol{s},\varepsilon)}
{q(\boldsymbol{\theta})}d\boldsymbol{\theta}
\nonumber\\
&= \int q(\boldsymbol{\theta})
\ln\frac{p(\boldsymbol{y}|\boldsymbol{s},\varepsilon)
p(\boldsymbol{s})p(\varepsilon)}
{q(\boldsymbol{\theta})}d\boldsymbol{\theta} \label{Lq-arbitrary}
\end{align}
where
\begin{align}
p(\boldsymbol{y}|\boldsymbol{s},\varepsilon)
&=\big(\frac{\varepsilon^{\frac{1}{2}}}{\sqrt{2\pi}}\big)^{N_r}
\exp\big\{-\frac{\varepsilon
\|\boldsymbol{y}-\boldsymbol{A}\boldsymbol{s}\|_2^2}{2}\big\}
\label{O1}
\end{align}
in which $\varepsilon$ still follows a Gamma distribution as
assumed in (\ref{prior-epsilon}).

Similar to what we did in the previous section, we resort to a
relaxed ELBO to facilitate the variational inference. Let
$\boldsymbol{T}\succcurlyeq\boldsymbol{\Sigma}^2$. Invoking Lemma
\ref{lemma1}, a lower bound on $p(\boldsymbol{y}|\boldsymbol{s},
\varepsilon)$ can be obtained as
\begin{align}
p(\boldsymbol{y}|\boldsymbol{s},\varepsilon)&=
\big(\frac{\varepsilon^{\frac{1}{2}}}{\sqrt{2\pi}}\big)^{N_r}
\exp\Big\{-\frac{\varepsilon}{2}\|\boldsymbol{y}-\boldsymbol{A}\boldsymbol{s}\|_2^2\Big\}
\nonumber\\
&\geqslant\big(\frac{\varepsilon^{\frac{1}{2}}}{\sqrt{2\pi}}\big)^{N_r}
\exp\Big\{-\frac{\varepsilon}{2}g(\boldsymbol{s},\boldsymbol{z})\Big\}
\triangleq
F(\boldsymbol{y},\boldsymbol{s},\varepsilon,\boldsymbol{z})
\label{g-y-inequality-2}
\end{align}
where
\begin{align}
 g(\boldsymbol{s},\boldsymbol{z})
 & \triangleq\|\boldsymbol{y}-\boldsymbol{A}\boldsymbol{z}\|_2^2+
2\Re\big\{(\boldsymbol{s}-\boldsymbol{z})^H\boldsymbol{A}^H(\boldsymbol{A}\boldsymbol{z}-\boldsymbol{y})\big\}
 +\nonumber\\
 &\quad\quad
 (\boldsymbol{s}-\boldsymbol{z})^H\boldsymbol{T}(\boldsymbol{s}-\boldsymbol{z})
 \end{align}
and the inequality becomes equality when $\boldsymbol{T}=
\boldsymbol{\Sigma}^2$. Here, we consider a specialized form of
$\boldsymbol{T}\triangleq
\boldsymbol{\overline{T}}+\frac{\delta}{\varepsilon}\boldsymbol{I}$,
where the first term is a diagonal matrix independent of
$\varepsilon$ and the second term is a function of $\varepsilon$.
We will show that such an expression of $\boldsymbol{T}$ leads to
an LMMSE-like estimator for the LE step.

Utilizing (\ref{g-y-inequality-2}), a relaxed ELBO can be obtained
as
\begin{align}
L(q)\geqslant \tilde{L}(q,\boldsymbol{z})&=\int
q(\boldsymbol{\theta})\ln\frac{G(\boldsymbol{y},\boldsymbol{s},\varepsilon,\boldsymbol{z})}
{q(\boldsymbol{\theta})}d\boldsymbol{\theta} \label{L-cost2}
\end{align}
where
\begin{align}
G(\boldsymbol{y},\boldsymbol{s},\varepsilon,\boldsymbol{z})=
F(\boldsymbol{y},\boldsymbol{s},\varepsilon,\boldsymbol{z})
p(\boldsymbol{s})p(\varepsilon)
\end{align}
Similarly, we aim to maximize the relaxed ELBO with respect to
$q(\boldsymbol{\theta})$ and the parameter $\boldsymbol{z}$, which
leads to a variational EM algorithm. In the E-step, given
$\boldsymbol{z}$, the approximate posterior distribution of
$\boldsymbol{\theta}\triangleq \{\boldsymbol{s},\varepsilon\}$ is
updated. In the M-step, the parameter $\boldsymbol{z}$ can be
updated by maximizing $\tilde{L}(q,\boldsymbol{z})$ with
$q(\boldsymbol{\theta})$ fixed.  Details of the variational
EM algorithm are provided below.

\textbf{A. E-step}

1) \emph{Update of $q_{\boldsymbol{s}}(\boldsymbol{s})$}: The
variational distribution $q_{\boldsymbol{s}}(\boldsymbol{s})$ can
be calculated as
\begin{align}
 \ln q_{\boldsymbol{s}}(\boldsymbol{s})&\propto
 \big\langle\ln{G}(\boldsymbol{y},\boldsymbol{s},\varepsilon,\boldsymbol{z})
 \big\rangle_{q_{\varepsilon}(\varepsilon)}
 \nonumber\\
 &\propto
 \big\langle
 \ln F(\boldsymbol{y},\boldsymbol{s},\varepsilon,\boldsymbol{z})+ \ln
p(\boldsymbol{s})\big\rangle_{q_{\varepsilon}(\varepsilon)}
 \nonumber\\
 &\propto
 \big\langle
 -\frac{\varepsilon}{2} g(\boldsymbol{s},\boldsymbol{z})
 +\ln p(\boldsymbol{s})\big\rangle_{q_{\varepsilon}(\varepsilon)}
\nonumber\\
 &\propto
 \Big\langle
 -\frac{\varepsilon}{2}\Big(
2\Re\big\{(\boldsymbol{s}-\boldsymbol{z})^H\boldsymbol{A}^H(\boldsymbol{A}\boldsymbol{z}-\boldsymbol{y})\big\}
 +
 \nonumber\\
 &\quad\quad(\boldsymbol{s}-\boldsymbol{z})^H\boldsymbol{T}(\boldsymbol{s}-\boldsymbol{z})
  \Big)
   +\ln
p(\boldsymbol{s})\Big\rangle_{q_{\varepsilon}(\varepsilon)}
\nonumber\\
 &\overset{(a)}{\propto}
 \ln N(\boldsymbol{s};\boldsymbol{\bar{r}},\boldsymbol{\bar{\Phi}})
  +\ln p(\boldsymbol{s})
\end{align}
where in $(a)$, we have
\begin{align}
 \boldsymbol{\bar{\Phi}} &=\frac{1}{\langle\varepsilon\rangle}
 \langle\boldsymbol{T}^{-1}\rangle
\\
 \boldsymbol{\bar{r}} &=\left\langle\boldsymbol{T}^{-1}\Big(
 \boldsymbol{A}^{H}\boldsymbol{y}+\boldsymbol{T}
 \boldsymbol{z}-\boldsymbol{A}^{H}
 \boldsymbol{A}\boldsymbol{z}\Big)\right\rangle
 \nonumber\\
 &=\boldsymbol{z}+\langle\boldsymbol{T}^{-1}\rangle
 \boldsymbol{A}^{H}\Big(
 \boldsymbol{y}-\boldsymbol{A}\boldsymbol{z}\Big)
 \label{r-bar-1}
\end{align}
in which
\begin{align}
\langle\boldsymbol{T}^{-1}\rangle =\boldsymbol{\overline{T}}^{-1}+
\frac{\langle\varepsilon\rangle}{\delta}\boldsymbol{I}
\end{align}
Note that the approximate posterior distribution
$q_{\boldsymbol{s}}(\boldsymbol{s})$ is difficult to obtain
because the prior $p(\boldsymbol{s})$ has an intractable
expression. To address this issue, we, instead, examine the
posterior distribution $q_{\boldsymbol{x}}(\boldsymbol{x})$.
Recalling that $\boldsymbol{x}=\boldsymbol{V}\boldsymbol{s}$, we
have
\begin{align}
\ln q_{\boldsymbol{x}}(\boldsymbol{x})&\propto
 \ln N(\boldsymbol{x};\boldsymbol{r},\boldsymbol{\Phi})
  +\ln
p(\boldsymbol{x})
\end{align}
where
\begin{align}
\boldsymbol{r}&=\boldsymbol{V}\boldsymbol{\bar{r}} \label{eqn2}
\\
\boldsymbol{\Phi}&=\boldsymbol{V}\boldsymbol{\bar{\Phi}}\boldsymbol{V}^H,
 \label{r-phi-1}
\end{align}
To calculate $q_{\boldsymbol{x}}(\boldsymbol{x})$, we neglect the
cross-correlation among entries of $\boldsymbol{x}$, i.e., treat
the off-diagonal entries of $\boldsymbol{\Phi}$ as zeros. This
trick helps obtain an analytical approximate posterior
distribution $q_{\boldsymbol{x}}(\boldsymbol{x})$, and meanwhile
has been empirically proven effective. Let $\Phi_i$ represent the
$i$th diagonal element of $\boldsymbol{\Phi}$. The first and
second moments of $x_i$ can be updated as
\begin{align}
\langle x_i\rangle&=E\{x_i;r_i,\Phi_i\}=\frac{\sum_{i}x_i N(x_i;
r_i,\Phi_{i}) p(x_i)}{\sum_{i}N(x_i;r_i,\Phi_{i}) p(x_i)}
  , \quad\forall i
  \nonumber\\
\langle x_i^2\rangle&=E\{x_i^2;r_i,\Phi_i\}=\frac{\sum_{i}x_i^2
N(x_i; r_i,\Phi_{i}) p(x_i)}{\sum_{i}N(x_i;r_i,\Phi_{i}) p(x_i)}
   , \quad\forall i
  \label{x-est-arbitrary}
\end{align}
Thus we arrive at $q_{\boldsymbol{x}}(\boldsymbol{x})$ follows a
Gaussian distribution with its mean $\langle
\boldsymbol{x}\rangle=[\langle
x_1\rangle\phantom{0}\ldots\phantom{0}\langle x_{N_t}\rangle]^T$
and covariance matrix
$\boldsymbol{\Sigma}_{\boldsymbol{x}}=\text{diag}(\Sigma_{x_1},\ldots,\Sigma_{x_{N_t}})$,
where
\begin{align}
  \Sigma_{x_i}&=\langle x_i^2\rangle-\langle x_i\rangle^2
   , \quad\forall i
\end{align}

Since $q_{\boldsymbol{x}}(\boldsymbol{x})$ follows a Gaussian
distribution, the approximate posterior of
$q_{\boldsymbol{s}}(\boldsymbol{s})$ also follows a Gaussian
distribution, and its mean and covariance matrix can be readily
obtained as
\begin{align}
 \langle\boldsymbol{s}\rangle=\boldsymbol{V}^H\langle\boldsymbol{x}\rangle
 \label{s-est2}
\end{align}
\begin{align}
 \boldsymbol{\Sigma}_{\boldsymbol{s}}=\boldsymbol{V}^H\boldsymbol{\Sigma}_{\boldsymbol{x}}\boldsymbol{V}
\end{align}

2) \emph{Update of $q_\varepsilon(\varepsilon)$}: The variational
distribution $q_\varepsilon(\varepsilon)$ can be calculated as
\begin{align}
\ln q_\varepsilon(\varepsilon)&\propto \langle
\ln{G}(\boldsymbol{y},\boldsymbol{s},\varepsilon,\boldsymbol{z})
\rangle_{q_{\boldsymbol{s}}(\boldsymbol{s})}
\nonumber\\
 &\propto
 \big\langle
 \ln F(\boldsymbol{y},\boldsymbol{s},\varepsilon,\boldsymbol{z})+\ln
p(\varepsilon)\big\rangle_{q_{\boldsymbol{s}}(\boldsymbol{s})}
\nonumber\\
 &\propto
 \Big(a-1+\frac{N_r}{2}\Big)\ln \varepsilon
 -\Big(\frac{1}{2}\langle \bar{g}(\boldsymbol{s},\boldsymbol{z})\rangle+b\Big)\varepsilon
\end{align}
in which
\begin{align}
\langle \bar{g}(\boldsymbol{s},\boldsymbol{z})\rangle
 =&
 \|\boldsymbol{y}-\boldsymbol{A}\boldsymbol{z}\|_2^2+
2\Re\{(\langle\boldsymbol{s}\rangle-\boldsymbol{z})^H\boldsymbol{A}^H
(\boldsymbol{A}\boldsymbol{z}-\boldsymbol{y})\}+
\nonumber\\
& \big(\langle\boldsymbol{s}\rangle-\boldsymbol{z}\big)^H
\boldsymbol{\overline{T}}\big(\langle\boldsymbol{s}\rangle-\boldsymbol{z}\big)
 +\text{Tr}\Big(\boldsymbol{\overline{T}}\boldsymbol{\Sigma}_{\boldsymbol{s}}
 \Big)
 \label{eps-a-est2}
\end{align}
Thus the hyperparameter $\varepsilon$ follows a Gamma
distribution, i.e.
\begin{align}
q_{\varepsilon}(\varepsilon)=\text{Gamma}(\varepsilon;\tilde{a},\tilde{b})
\end{align}
where $\tilde{a}$ and $\tilde{b}$ are given by
\begin{align}
\tilde{a}=a+\frac{N_r}{2}, \quad \tilde{b}=b+\frac{1}{2}\langle
\bar{g}(\boldsymbol{s},\boldsymbol{z})\rangle \label{eps-a-est0}
\end{align}
Also, we have
\begin{align}
\langle\varepsilon\rangle=\frac{\tilde{a}}{\tilde{b}}
 \label{eps-a-est1}
\end{align}

\textbf{B. M-step}

Substituting $q(\boldsymbol{\theta};\boldsymbol{z}^{\text{old}})$
into $\tilde{L}(q,\boldsymbol{z})$, the parameter $\boldsymbol{z}$
can be optimized via
\begin{align}
\boldsymbol{z}^{\text{new}}&=\arg
\max_{\boldsymbol{z}}\langle\ln{G}(\boldsymbol{y},\boldsymbol{s},\varepsilon,\boldsymbol{z})
\rangle_{q(\boldsymbol{\theta};\boldsymbol{z}^{\text{old}})}
\end{align}
Setting the derivative of the logarithm function to zero yields
\begin{align}
\frac{\partial
\langle\ln{G}(\boldsymbol{y},\boldsymbol{\theta},\boldsymbol{z})
\rangle_{q(\boldsymbol{\theta};\boldsymbol{z}^{\text{old}})}}{\partial
\boldsymbol{z}} =\langle\varepsilon\rangle\big(\boldsymbol{\overline{T}}+
\frac{\delta}{\langle\varepsilon\rangle}\boldsymbol{I}
 -\boldsymbol{\Sigma}^2\big)\big(\langle\boldsymbol{s}\rangle-\boldsymbol{z}\big)
=\boldsymbol{0}
\end{align}
Since $\langle\varepsilon\rangle>0$ and $\boldsymbol{\overline{T}}+
\frac{\delta}{\langle\varepsilon\rangle}\boldsymbol{I}\succeq
\boldsymbol{\Sigma}^2$, the solution of $\boldsymbol{z}$ is given
by
\begin{align}
\boldsymbol{z}^{\text{new}}=\langle\boldsymbol{s}\rangle
\label{w-est}
\end{align}

\subsection{Summary of The Improved IFVB Detector} For the sake of clarity,
the improved IFVB detector can be summarized as an iterative
algorithm with each iteration consisting of a linear estimation
step and a nonlinear estimation step:
\begin{align}
\quad  \text{LE} : \quad\quad~ \boldsymbol{\bar{r}}_t
&=\boldsymbol{s}_t+\boldsymbol{T}_t^{-1}
 \boldsymbol{A}^{H}(
 \boldsymbol{y}-\boldsymbol{A}\boldsymbol{s}_t)  \label{LE-arbitrary-a}
\\
 \boldsymbol{r}_t&=\boldsymbol{V}\boldsymbol{\bar{r}}_t \label{LE-arbitrary-b}
 \\
\quad  \text{NLE} : ~~ x_{i, t+1}&=E\{x_i;r_{i,t}, \Phi_{i,t}\},
\quad \forall i \label{NLE-arbitrary-a}
\\
  \boldsymbol{s}_{t+1} &=\boldsymbol{V}^H\boldsymbol{x}_{t+1} \label{NLE-arbitrary-b}
\end{align}
 where $\boldsymbol{T}_t\triangleq
\boldsymbol{\overline{T}}+\frac{\delta}{\varepsilon_t}\boldsymbol{I}$, and
 $r_{i,t}$ and $\Phi_{i,t}$ denote the $i$th elements of
$\boldsymbol{r}_t$ and $\boldsymbol{\Phi}_t$ , respectively. The
linear estimator in equations (\ref{LE-arbitrary-a}) and
(\ref{LE-arbitrary-b}) derives from (\ref{r-bar-1}) and
(\ref{eqn2}). The nonlinear estimator in equations
(\ref{NLE-arbitrary-a}) and (\ref{NLE-arbitrary-b}) follows from
(\ref{x-est-arbitrary}) and (\ref{s-est2}). In
(\ref{NLE-arbitrary-a}), the variance $\Phi_{i,t}$ can be obtained
as
\begin{align}
\Phi_{i,t}=\boldsymbol{v}_i^r\boldsymbol{\bar{\Phi}}_t(\boldsymbol{v}_i^r)^H,
\quad \forall i
\end{align}
where $\boldsymbol{v}_i^r$ denotes the $i$th row of
$\boldsymbol{V}$, and $\boldsymbol{\bar{\Phi}}_t
=\frac{1}{\varepsilon_t}\boldsymbol{T}_t^{-1}$. Here
$\varepsilon_{t+1}$ is updated as
\begin{align}
\varepsilon_{t+1}=\frac{\tilde{a}}{\tilde{b}_{t+1}}=\frac{a+\frac{N_r}{2}}{b+\frac{1}{2}
\bar{g}(\boldsymbol{s}_{t+1},\boldsymbol{s}_t)}
\end{align}
where
\begin{align}
 \bar{g}(\boldsymbol{s}_{t+1},\boldsymbol{s}_{t})
 =&
 \|\boldsymbol{y}-\boldsymbol{A}\boldsymbol{s}_{t}\|_2^2+
2\Re\big\{(\boldsymbol{s}_{t+1}-\boldsymbol{s}_{t})^H\boldsymbol{A}^H
(\boldsymbol{A}\boldsymbol{s}_{t}-\boldsymbol{y})\big\}
\nonumber\\
& +\big(\boldsymbol{s}_{t+1}-\boldsymbol{s}_{t}\big)^H
\boldsymbol{\overline{T}}\big(\boldsymbol{s}_{t+1}-\boldsymbol{s}_{t}\big)
 +\text{Tr}\big(\boldsymbol{\overline{T}}\boldsymbol{\Sigma}_{\boldsymbol{s}}^{t+1}
 \big) \label{g-s}
\end{align}
in which
\begin{align}
 \boldsymbol{\Sigma}_{\boldsymbol{s}}^{t+1}=\boldsymbol{V}^H
 \boldsymbol{\Sigma}_{\boldsymbol{x}}^{t+1}\boldsymbol{V}
 \label{Sigma-s-t}
 \end{align}
and the $i$th diagonal element of the diagonal matrix
$\boldsymbol{\Sigma}_x^{t+1}$ can be calculated as
\begin{align}
\Sigma_{x_i}^{t+1}=
E\{x_{i}^2;r_{i,t},\Phi_{i,t}\}-\big[E\{x_{i};r_{i,t},\Phi_{i,t}\}\big]^2,\quad
\forall i  \label{x-Sigma-arbitrary}
\end{align}

\emph{Discussions:} For the Improved IFVB Detector, a specialized
form of $\boldsymbol{T}\triangleq
\boldsymbol{\overline{T}}+\frac{\delta}{\varepsilon}\boldsymbol{I}$
is considered. The reason is that such a specialized form will
lead to an LMMSE-like estimator for the LE step, i.e.
$\boldsymbol{\bar{r}}_t=\boldsymbol{s}_t+\boldsymbol{T}_t^{-1}
\boldsymbol{A}^{H}
(\boldsymbol{y}-\boldsymbol{A}\boldsymbol{s}_t)$ (see
(\ref{LE-arbitrary-a})). Here, $\boldsymbol{\bar{r}}_t$ denotes
the current estimate of $\boldsymbol{s}$ in the LE step, and
$\boldsymbol{s}_t$ denotes the previous estimate of
$\boldsymbol{s}$. If we treat $\boldsymbol{s}-\boldsymbol{s}_t$ as
a random variable following a Gaussian distribution with zero mean
and covariance matrix $\delta^{-1}\boldsymbol{I}$, and let
$\boldsymbol{\overline{T}}=\boldsymbol{\Sigma}^2=\boldsymbol{A}^H\boldsymbol{A}$.
It can be readily verified that
$\boldsymbol{\bar{r}}_t-\boldsymbol{s}_t=\boldsymbol{T}_t^{-1}
\boldsymbol{A}^{H}
(\boldsymbol{y}-\boldsymbol{A}\boldsymbol{s}_t)$ is in fact an
LMMSE estimator of $\boldsymbol{s}-\boldsymbol{s}_t$ given the
observation: $\boldsymbol{y}-\boldsymbol{A}\boldsymbol{s}_t=
\boldsymbol{A}(\boldsymbol{s}-\boldsymbol{s}_t)+\boldsymbol{n}$.
It is known that the LMMSE estimator can compensate for the noise
and effectively improve the estimation performance. This is the
reason why we choose such a specialized form of $\boldsymbol{T}$.

\begin{figure*}[t]
\setlength{\abovecaptionskip}{0pt}
\setlength{\belowcaptionskip}{0pt} \centering
\includegraphics[width=12cm]{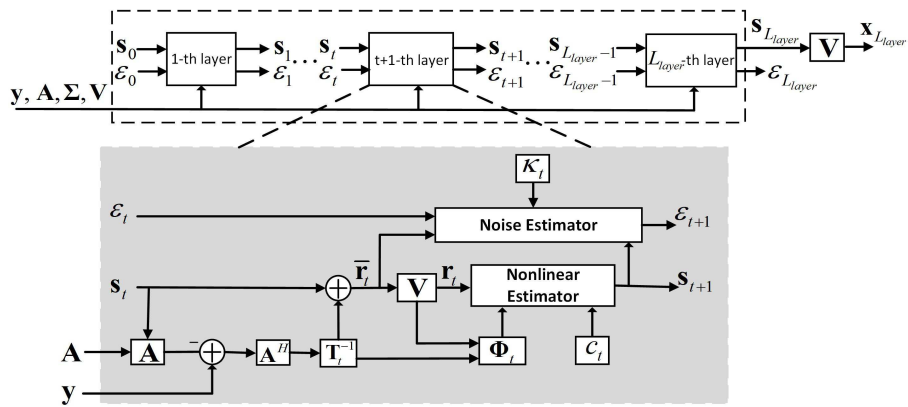}
\caption{Block diagram of the Improved-VBINet detector.}
\label{correlated-network}
\end{figure*}

\section{VBINet Detector for Correlated Channels}  \label{section-ImprovedIFBLNet}
\label{IFBLNet-chapter} In this section, we propose a model-driven
DL detector (referred to as Improved-VBINet) by unfolding the
improved IFVB detector. The network architecture of
Improved-VBINet is illustrated in Fig. \ref{correlated-network},
in which the network consists of $L_{\text{layer}}$ cascade
layers. The trainable parameters of all layers are denoted as
$\boldsymbol{\Omega}\triangleq\big\{\boldsymbol{\Psi},
\{\delta_t\}_{t=1}^{L_\text{layer}},
\{c_t\}_{t=1}^{L_\text{layer}},\{\kappa_t\}_{t=1}^{L_\text{layer}}\big\}$,
where the diagonal matrix $\boldsymbol{\Psi}$ and $\delta_t$ are
related to the diagonal matrix $\boldsymbol{T}_t$. For the
$(t+1)$th layer, the input includes $\boldsymbol{y}$,
$\boldsymbol{A}$, $\boldsymbol{\Sigma}$, $\boldsymbol{V}$,
$\boldsymbol{s}_{t}$, and $\varepsilon_{t}$, where
$\boldsymbol{s}_{t}$ and $\varepsilon_{t}$ denote the $t$th
layer's estimate of the signal and the noise variance,
respectively. Given $\{\boldsymbol{y},\boldsymbol{A},
\boldsymbol{\Sigma},
\boldsymbol{V},\boldsymbol{s}_{t},\varepsilon_{t}\}$, each layer
performs the following updates:
\begin{align}
\quad & \text{LE} : \quad\quad\quad~ \boldsymbol{\bar{r}}_t
=\boldsymbol{s}_t+\boldsymbol{T}_t^{-1}
 \boldsymbol{A}^{H}(
 \boldsymbol{y}-\boldsymbol{A}\boldsymbol{s}_t) \label{LENet-arbitrary-a},
\\
\quad & \quad\quad\quad\quad\quad~
\boldsymbol{r}_t=\boldsymbol{V}\boldsymbol{\bar{r}}_t \label{LENet-arbitrary-b}
\\
\quad & \text{NLE} : \quad
 x_{i, t+1} = E\{x_i;r_{i,t}, \Phi_{i,t}\},
\quad \forall i
 \label{NLENet-arbitrary-a}
 \\
\quad & \quad\quad\quad~~~~ \boldsymbol{s}_{t+1} =
c_t\boldsymbol{V}^H\boldsymbol{x}_{t+1} +(1-c_t)\boldsymbol{s}_t
\label{NLENet-arbitrary-b}
 \\
\quad & \text{Update of $\varepsilon_{t+1}$} :\quad
\varepsilon_{t+1}=\frac{a+\frac{N_r}{2}}{b+\frac{1}{2}
g_{\text{net}}(\boldsymbol{s}_{t+1},\boldsymbol{s}_t)}
\label{NLENet-eps}
\end{align}
where
\begin{align}
& \quad\quad\quad\quad
\boldsymbol{T}_t\triangleq\boldsymbol{\overline{T}}+
\frac{\delta_t}{\varepsilon_t}\boldsymbol{I}, \quad
\boldsymbol{\overline{T}}\triangleq\boldsymbol{\Sigma}^2+\boldsymbol{\Psi}^2
\\
 & \quad\quad\quad\quad
\boldsymbol{\bar{\Phi}}_t=\frac{1}{\varepsilon_t}\boldsymbol{T}_t^{-1}
,  \quad\quad\quad\quad\quad\quad\\
 & \quad\quad\quad\quad
\Phi_{i,t}=\boldsymbol{v}_i^r\boldsymbol{\bar{\Phi}}_t(\boldsymbol{v}_i^r)^H,
\quad \forall i \quad\quad\quad\quad\quad\quad
\\
&
g_{\text{net}}(\boldsymbol{s}_{t+1},\boldsymbol{s}_{t})
\nonumber\\
 &=
 \|\boldsymbol{y}-\boldsymbol{A}\boldsymbol{s}_{t}\|_2^2+
2\Re\big\{(\boldsymbol{s}_{t+1}-\boldsymbol{s}_{t})^H\boldsymbol{A}^H
(\boldsymbol{A}\boldsymbol{s}_{t}-\boldsymbol{y})\big\}
\nonumber\\
& + \big(\boldsymbol{s}_{t+1}-\boldsymbol{s}_{t}\big)^H
\boldsymbol{\overline{T}}\big(\boldsymbol{s}_{t+1}-\boldsymbol{s}_{t}\big)
 +  \text{Tr}\Big(\boldsymbol{\overline{T}}
 \boldsymbol{\Sigma}_{\text{net}}^{t+1}
 \Big),   \quad\quad\quad
 \label{g-net-arbitrary}
 \\
 &\quad\quad\quad\quad \boldsymbol{\Sigma}_{\text{net}}^{t+1}=c_t^2 \kappa_t^2\boldsymbol{I}.
 \label{var-net-arbitrary}
 \end{align}

We see that the update formulas
(\ref{LENet-arbitrary-a})--(\ref{var-net-arbitrary}) are similar
to the improved IFBL detector's update formulas
(\ref{LE-arbitrary-a})--(\ref{x-Sigma-arbitrary}). Nevertheless,
there are several major differences. Firstly, in the LE step,
$\boldsymbol{T}_t$ which is calculated according to certain rules
is replaced by a trainable diagonal matrix $\boldsymbol{T}_t$ due
to the use of the learnable variables $\big\{\boldsymbol{\Psi},
\{\delta_t\}_{t=1}^{L_\text{layer}}\big\}$. Specifically, the
linear estimator module turns out to have a form similar to the
LMMSE estimator, i.e.
\begin{align}
  \boldsymbol{\bar{r}}_t
  &=\boldsymbol{s}_t+\boldsymbol{T}_t^{-1}
 \boldsymbol{A}^{H}\Big(
 \boldsymbol{y}-\boldsymbol{A}\boldsymbol{s}_t\Big)
 \nonumber\\
  &
 =\boldsymbol{s}_t+\Big(\boldsymbol{\overline{T}}
 +\frac{\delta_t}{\varepsilon_t}\boldsymbol{I}\Big)^{-1}
 \boldsymbol{A}^{H}\Big(
 \boldsymbol{y}-\boldsymbol{A}\boldsymbol{s}_t\Big)
 \label{p-est-1}
  \end{align}
Secondly, in the NLE step, $\boldsymbol{s}_{t+1}$ is updated by
applying an adaptive damping scheme with a learnable parameter
$c_t$. Due to the use of the term $(1-c_t)\boldsymbol{s}_t$ in the
NLE step, the output of the NLE is the mean of a new variable
$\boldsymbol{s}_{\text{net}}$ defined as
\begin{align}
\boldsymbol{s}_{\text{net}}= c_t \boldsymbol{s}+ (1-c_t)
\boldsymbol{s}_t
\end{align}
Recalling that $\boldsymbol{s}=\boldsymbol{V}^H\boldsymbol{x}$,
thus the mean of $\boldsymbol{s}_{\text{net}}$ is given by
\begin{align}
\boldsymbol{s}_{t+1}=\langle\boldsymbol{s}_{\text{net}}\rangle=
c_t \boldsymbol{V}^H E\{\boldsymbol{x};\boldsymbol{r}_t,
\boldsymbol{\Phi}_t\}+ (1-c_t) \boldsymbol{s}_t
\end{align}
and the covariance matrix $\boldsymbol{\Sigma}_{\text{net}}^{t+1}$
of the random variable $\boldsymbol{s}_{\text{net}}$ is given by
\begin{align}
\boldsymbol{\Sigma}_{\text{net}}^{t+1} = c_t^2
\boldsymbol{\Sigma}_{s}^{t+1}=c_t^2\boldsymbol{V}^H
 \boldsymbol{\Sigma}_{\boldsymbol{x}}^{t+1}\boldsymbol{V}
 \label{snet-var}
\end{align}
Here $\boldsymbol{\Sigma}_{\boldsymbol{x}}^{t+1}$ is a diagonal
matrix whose diagonal elements can be calculated according to
(\ref{x-Sigma-arbitrary}). To further reduce the complexity, we
approximate $\boldsymbol{\Sigma}_{\boldsymbol{x}}^{t+1}$ as
$\kappa_t^2\boldsymbol{I}$, where $\kappa_t$ is a learnable
parameter. Thus we have
\begin{align}
\boldsymbol{\Sigma}_{\text{net}}^{t+1}=
 c_t^2 \kappa_t^2\boldsymbol{I}
\end{align}
This explains how (\ref{var-net-arbitrary}) is derived.

\emph{Remark 1:} We see that the total number of learnable
parameters of the improved-VBINet is $|\boldsymbol{\Omega}|= N_t +
3 L_{\text{layer}}$, as each layer shares the same learnable
variable $\boldsymbol{\Psi}$ but $\{\delta_t, c_t,\kappa_t\}$ are
generally different for different layers. The total number of
learnable parameters of the Improved-VBINet is comparable to that
of OAMPNet.

\subsection{Computational Complexity Analysis} We
discuss the computational complexity of our proposed VBINet
detectors. Specifically, for the i.i.d. Gaussian channels, at each
layer the complexity of the proposed VBINet detector is dominated
by the multiplication of an $N_t\times N_r$ matrix and an
$N_r\times 1$ vector. Therefore the overall complexity of the
proposed VBINet is of order $\mathcal{O}(N_r N_t
L_{\text{layer}})$. For the correlated channels, the proposed
Improved-VBINet detector requires to conduct a truncated SVD
decomposition of an $N_r\times N_t$ matrix in the initialization
stage. Therefore the overall complexity of the proposed
Improved-VBINet detector is of order $\mathcal{O}(N_r N_t
L_{\text{layer}}+N_t^3)$. Similar to the proposed VBINet, the
complexity of MMNet is of order $\mathcal{O}(N_r N_t
L_{\text{layer}})$. As for the OAMPNet, it requires to perform an
inverse of an $N_r\times N_r$ matrix per layer, which results in a
complexity of $\mathcal{O}(N_r^3 L_{\text{layer}})$ in total.
Since we usually have $N_r\gg N_t$ for massive MIMO detection, the
complexity of the OAMPNet is much higher than our proposed method.
This is also the reason why, for large-scale scenarios, training
and test of OAMPNet becomes computationally prohibitive.

\section{Simulation Results}    \label{section-simulation}
In this section, we compare the proposed VBINet-based
detector\footnote{Codes are available at
https://www.junfang-uestc.net/codes/VBINet.rar} with
state-of-the-art methods for i.i.d. Gaussian, correlated Rayleigh
and realistic 3GPP MIMO channel matrices. In the following, we
first discuss the implementation details of the proposed method
and other competing algorithms.

\subsection{Implementation Details}
During training, the learnable parameters are optimized using
stochastic gradient descent. Thanks to the auto-differentiable
machine-learning frameworks such as TensorFlow
\cite{AbadiAgarwal16} and PyTorch \cite{PaszkeGross17}, the
learnable parameters can be easily trained when the loss function
is determined. In our experiments, the loss function used for
training is given by
\begin{align}
f_{loss}=\frac{1}{L_{\text{layer}}}\sum_{t=1}^{L_{\text{layer}}}
\|\boldsymbol{x}_t-\boldsymbol{x}\|_2^2
\end{align}
in which $\|\boldsymbol{x}_t-\boldsymbol{x}\|_2^2$ denotes the
square error between the output of the $t$-th layer and the true
signal $\boldsymbol{x}$.

In addition to the classical detectors such as the ZF and LMMSE,
we compare our proposed method with the iterative algorithm OAMP
and the following two state-of-the-art DL-based detectors: OAMPNet
\cite{HeWen20} and MMNet \cite{KhaniAlizadeh20}. The ML detector
is also included to provide a performance upper bound for all
detectors. Given a set of training samples, the neural network is
trained by the Adam optimizer \cite{KingmaBa14} in the PyTorch
framework.

In our simulations, the training and test samples
$\{\boldsymbol{y},\boldsymbol{x}\}$ are generated according to the
model (\ref{data-model}), in which both the transmitted signal
$\boldsymbol{x}$, the channel $\boldsymbol{H}$, and the
observation noise $\boldsymbol{n}$ are randomly generated.
Specifically, the channel $\boldsymbol{H}$ is either generated
from an i.i.d. Gaussian distribution or generated according to a
correlated channel model as described below. The number of
training iterations is denoted as $N_{\text{iter}}$, and the batch
size for each
 iteration is denoted as $N_{\text{batch}}$. The signal-to-noise ratio
(SNR) is defined as
\begin{align}
\text{SNR}(\text{dB})= 10\log
\Big(\frac{E[\|\boldsymbol{H}\boldsymbol{x}\|_2^2]}{E[\|\boldsymbol{n}\|_2^2]}\Big)
\end{align}
The training samples in each training iteration includes samples
drawn from different SNRs within a certain range. For all
detectors, the output signal will be rounded to the closest point
on the discrete constellation set $\mathcal{C}$. Details of
different detectors are elaborated below.

$\bullet~\textbf{\text{ZF}}$: It is a simple yet classical decoder
which has a form as
$\boldsymbol{\hat{x}}=(\boldsymbol{H}^H\boldsymbol{H})^{-1}
\boldsymbol{H}^H\boldsymbol{y}$.

$\bullet~\textbf{\text{LMMSE}}$: The LMMSE detector takes a form
of
$\boldsymbol{\hat{x}}=\boldsymbol{H}^H\big(\boldsymbol{H}\boldsymbol{H}^H+
\frac{1}{\varepsilon\cdot
P_x}\boldsymbol{I}\big)^{-1}\boldsymbol{y}$. Here $P_x$ is the
average power of the QAM signal. In our simulation, we set
$P_x=1$.

$\bullet~\textbf{\text{IFVB}}$: The algorithm is summarized in
\ref{sum-IFBL}. The maximum number of iterations is set to
$L_{\text{layer}}=100$.

$\bullet~\textbf{\text{OAMP}}$: It is an iterative algorithm
developed based on the AMP algorithm. The maximum number of
iterations is set to $L_{\text{layer}}=100$.

$\bullet~\textbf{\text{OAMPNet}}$: OAMPNet is a DL-based
detector developed by unfolding the OAMP detector \cite{HeWen20}. In our
simulations, the number of layers of the OMAPNet is set to $10$,
and each layer has $4$ learnable variables.

$\bullet~\textbf{\text{MMNet-iid}}$: The $\text{MMNet-iid}$
\cite{KhaniAlizadeh20} is specifically designed for i.i.d.
Gaussian channels. In our simulations, the number of layers of the
MMNet-iid is set to $10$, and each layer has $2$ learnable
variables.

$\bullet~\textbf{\text{MMNet}}$:  The $\text{MMNet}$
\cite{KhaniAlizadeh20} can be applied to arbitrarily correlated
channel matrices, which needs to learn a matrix $\boldsymbol{A}_t$
for each layer, and the total number of learnable parameters is $2
N_t(N_r+1)$ per layer.

$\bullet~\textbf{\text{VBINet}}$: The architecture of our proposed
VBINet is shown in Section \ref{IFBL-iid-chapter}, which has
$L_{\text{layer}}$ layers and $N_t + L_{\text{layer}}$ learnable
variables.

$\bullet~\textbf{\text{Improved-VBINet}}$: Applicable for
arbitrarily correlated matrices, the proposed Improved-VBINet has
$L_{\text{layer}}$ layers and $N_t + 3 L_{\text{layer}}$ learnable
variables.

$\bullet~\textbf{\text{ML}}$: It is implemented by using a highly
optimized Mixed Integer Programming package Gurobi \cite{optimizationinc15}
\footnote{The code is shared by \cite{KhaniAlizadeh20}.}.

\begin{figure}[t]
\setlength{\abovecaptionskip}{0pt}
\setlength{\belowcaptionskip}{0pt} \centering
\includegraphics[width=7.1cm]{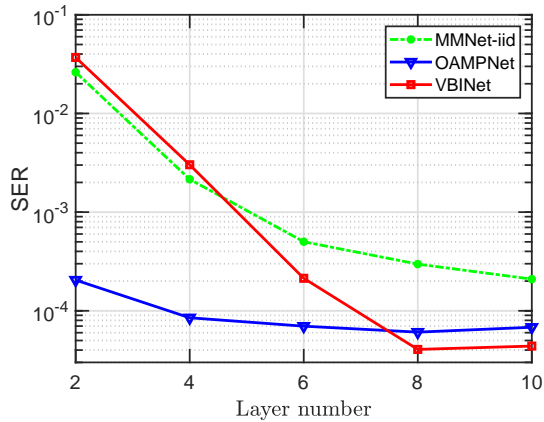}
\caption{SER vs. Number of layers number on i.i.d. Gaussian
channels.} \label{layer-num}
\end{figure}

\begin{figure}[t]
\setlength{\abovecaptionskip}{0pt}
\setlength{\belowcaptionskip}{0pt} \centering
\includegraphics[width=7.1cm]{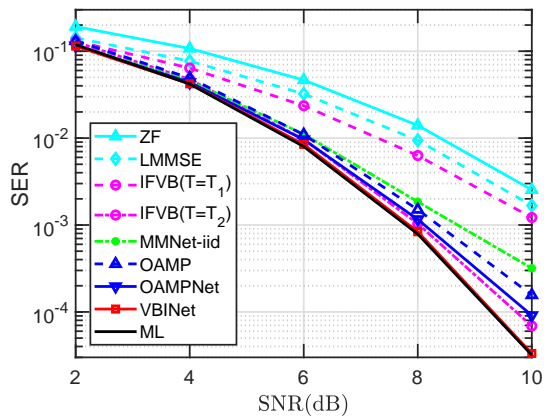}
\caption{SERs of respective schemes vs. SNR on i.i.d. Gaussian
channels.} \label{perfectCSI-IID-QPSK}
\end{figure}

\subsection{Results}
\subsubsection{Convergence Evaluation}
The convergence property of our proposed network is evaluated.
Fig. \ref{layer-num} depicts the symbol error rate (SER) of
different detectors as a function of the number of layers
$L_{\text{layer}}$ with QPSK modulation and i.i.d. Gaussian
channels, where we set $N_t=16$, $N_r=32$, $N_{\text{batch}}=500$,
and $N_{\text{iter}}=10^4$. We see that our proposed VBINet-based
detector converges in less than $10$ layers, and after
convergence, the proposed detector can achieve a certain amount of
performance improvement over the OAMPNet and MMNet-iid.

\subsubsection{I.I.D. Gaussian Channels}
We first examine the detection performance of our proposed
detector on i.i.d. Gaussian channels. We consider an offline
training mode, in which training is performed over randomly
generated channels and the performance is evaluated over another
set of randomly generated channels.

Fig. \ref{perfectCSI-IID-QPSK} plots the SER of respective
algorithms as a function of the SNR, where QPSK modulation is
considered, and we set $N_t=16$, $N_r=32$, $N_{\text{batch}}=500$,
and $N_{\text{iter}}=10^4$. It can be observed that the ML
detector achieves the best performance among all detectors, and
the SER gap between our proposed VBINet and the ML detector is
negligible. Moreover, the proposed VBINet outperforms the OAMPNet.
We also see that the VBINet presents a clear performance advantage
over the IFVB, which is because a more suitable $\boldsymbol{T}$
learned by VBINet will lead to better performance. For the IFVB,
we use two different choices of $\boldsymbol{T}$ to show the
impact of $\boldsymbol{T}$ on the detection performance.
Specifically, the first choice of $\boldsymbol{T}$ is set to
$\boldsymbol{T}_1=(\lambda_{\text{max}}(\boldsymbol{H}^H\boldsymbol{H})+
\epsilon)\boldsymbol{I}$ and the second choice of $\boldsymbol{T}$
is set to $\boldsymbol{T}_2=\boldsymbol{H}_d$, where
$\boldsymbol{H}_d$ is a diagonal matrix which takes the diagonal
entries of $\boldsymbol{H}^H\boldsymbol{H}$. We see that the
latter choice leads to much better performance, which indicates
that the performance of IFVB depends largely on the choice of
$\boldsymbol{T}$. This is also the motivation why we learn
$\boldsymbol{T}$ in the VBINet.

\subsubsection{Correlated Rayleigh Channels} We now examine
the performance of respective detectors on correlated Rayleigh
channels. The correlated Rayleigh channel is modeled as
\begin{align}
\boldsymbol{H}=\boldsymbol{R}_{H}^{\frac{1}{2}}\boldsymbol{G}_{H}
\boldsymbol{T}_{H}^{\frac{1}{2}}
\end{align}
where entries of $\boldsymbol{G}_{H}$ follow i.i.d. Gaussian
distribution with zero mean and unit variance,
$\boldsymbol{R}_{H}$ and $\boldsymbol{T}_{H}$ respectively denote
the receive and transmit correlation matrices. Note that the
receive and transmit correlation matrices can be characterized by
a correlation coefficient $\rho$ \cite{Loyka01}, and a larger
$\rho$ means that the channel matrix has a larger condition
number. Again, in this experiment, offline training is considered.

In Fig. \ref{Perfect-correlated-IID-QPSK}, we plot the SER of
respective algorithms as a function of SNR, where we set
$\rho=0.8$, $N_t=16$, $N_r=32$, $N_{\text{batch}}=500$, and
$N_{\text{iter}}=10^4$. It can be observed that all DL detectors
suffer a certain amount of performance degradation as the channel
becomes ill-conditioned. Besides, MMNet fails to work in the
offline training mode, where the test channels are different from
the channels generated during the training phase. This is because
in the LE step, the MMNet treats $\boldsymbol{A}_t$ as an entire
trainable matrix. As a result, the trained neural network can only
accommodate a particular channel realization. We also observe that
both our proposed Improved-VBINet and OAMPNet work well in the
offline training mode, and the proposed Improved-VBINet achieves
performance slightly better than the OAMPNet.

\begin{figure}[t]
\setlength{\abovecaptionskip}{0pt}
\setlength{\belowcaptionskip}{0pt} \centering
\includegraphics[width=7.1cm]{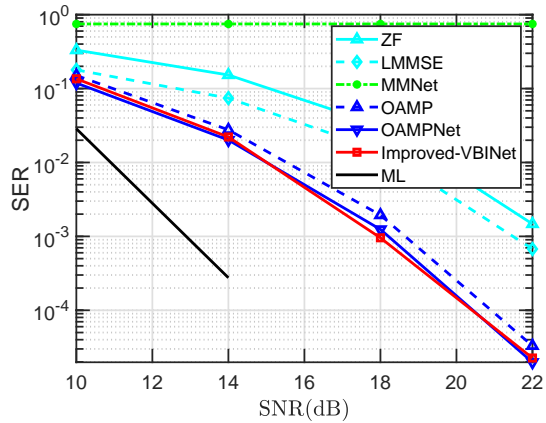}
\caption{SERs of respective detectors vs. SNR on correlated
Rayleigh channels.} \label{Perfect-correlated-IID-QPSK}
\end{figure}

\begin{figure}[t]
\setlength{\abovecaptionskip}{0pt}
\setlength{\belowcaptionskip}{0pt} \centering
\includegraphics[width=7.1cm]{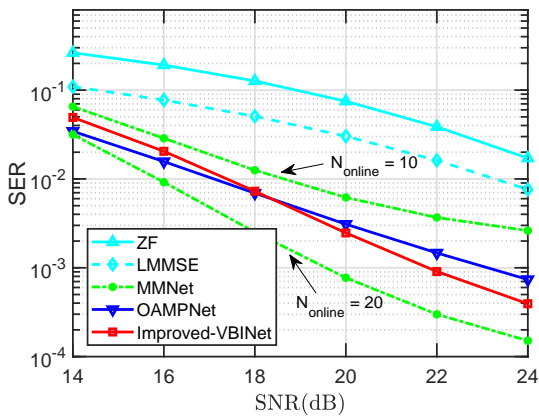}
\caption{SERs of respective detectors vs. SNR on realistic 3GPP
MIMO channels.} \label{3GPP_onlie_QPSK}
\end{figure}

\begin{figure}[t]
\setlength{\abovecaptionskip}{0pt}
\setlength{\belowcaptionskip}{0pt} \centering
\includegraphics[width=7.1cm]{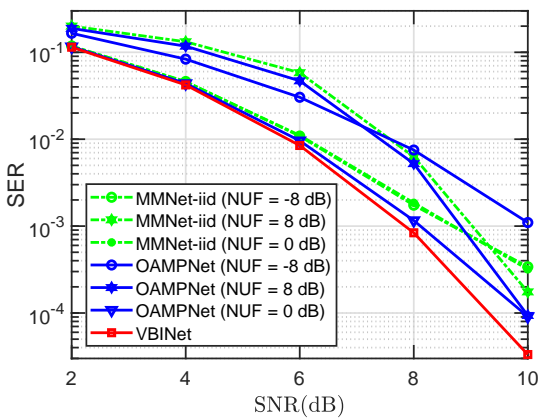}
\caption{SERs of respective detectors vs. SNR under different
NUFs, where i.i.d. Gaussian channels are considered.}
\label{untrueSNR-IID-QPSK-delta}
\end{figure}

\begin{figure}[t]
\setlength{\abovecaptionskip}{0pt}
\setlength{\belowcaptionskip}{0pt} \centering
\includegraphics[width=7.1cm]{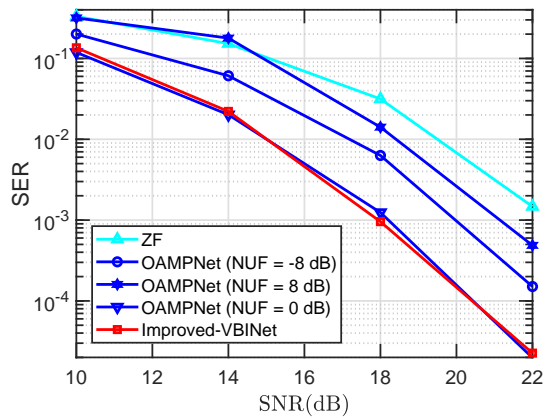}
\caption{SERs of respective detectors vs. SNR under different
NUFs, where correlated Rayleigh channels are considered.}
\label{untrueSNR-correlated-QPSK-delta}
\end{figure}

\subsubsection{3GPP MIMO Channels}
Next, we compare the performance of different detectors on
realistic 3GPP 3D MIMO channels which are generated by the
QuaDRiGa channel simulator \cite{JaeckelRaschkowski14, 3GPP}. The
parameters used to generate the channels are set the same as those
in \cite{KhaniAlizadeh20}, except that the bandwidth is set to
$1$MHz, the number of effective sub-carriers is set to $F=128$,
and the number of sequences is $2$. In this experiments, online
training is considered, where the model parameters are trained and
tested on the same realization of the channel. Specifically, we
first train the neural network associated with the first
subcarrier's channel. The trained model parameters are then used
as initial values of the neural network of the second subcarrier,
and the second subcarrier's network is employed as a start to
train the third subcarrier's network, and so on. Due to the
channel correlation among adjacent subcarriers, the rest
subcarriers' neural networks can be efficiently trained with a
small amount of data samples.

Fig. \ref{3GPP_onlie_QPSK} depicts the SER of respective
algorithms as a function of SNR, where we set $N_t=16$, $N_r=32$,
and the QPSK modulation is considered. For the first subcarrier,
the number of training iterations is $1000$ and the batch size is
$500$, while for all subsequent subcarriers, the number of online
training iterations is $10$ and the batch size is $500$. In Fig.
\ref{3GPP_onlie_QPSK}, we also report results of the MMNet when
the number of online training iterations is set to
$N_{\text{online}}=20$. We see that our proposed Improved-VBINet
achieves performance similar to the OAMPNet. Also, it can be
observed that MMNet achieves the best performance when a
sufficient amount of data is allowed to be used for training.
Nevertheless, the performance of MMNet degrades dramatically with
a reduced amount of training data. Note that MMNet needs to learn
an entire matrix $\boldsymbol{A}_t$ at each layer. With a large
number of learnable parameters, the MMNet can well approximate the
best detector when sufficient training is performed, whereas
incurs a significant amount of performance degradation when
training is insufficient.

\subsubsection{Noise Uncertainty}
Next, we examine the impact of noise variance uncertainty on the
performance of different DL-based detectors. As mentioned earlier,
different from our proposed VBINet which is capable of estimating
the noise variance automatically, both OAMPNet and MMNet require
the knowledge of noise variance for training and test. Note that
the noise variance can be assumed known during the training phase,
as we usually have access to the model parameters used to generate
the training data samples. But the same assumption does not hold
true for the test data. As a result, when evaluating the
performance on the test data, both OAMPNet and MMNet have to
replace the noise variance by its estimate. Let the estimated
noise variance be $1/\hat{\varepsilon}=\eta(1/\varepsilon)$. We
also define the noise uncertainty factor (NUF) as
$\text{NUF}\triangleq 10\log_{10}\eta$. In this experiment, an
offline training mode is considered.

In Fig. \ref{untrueSNR-IID-QPSK-delta}, we plot the SER of
respective algorithms as a function of SNR with different NUFs,
where i.i.d. Gaussian channels and QPSK modulation are considered,
and we set $N_t=16$, $N_r=32$, $N_{\text{batch}}=500$, and
$N_{\text{iter}}=10^4$. It can be observed that both MMNet-iid and
OAMPNet incur a considerable amount of performance loss when the
estimated noise variance deviates from the true one. The
performance gap between OAMPNet and VBINet becomes more pronounced
when the estimate of the noise variance is inaccurate. Fig.
\ref{untrueSNR-correlated-QPSK-delta} plots the SER of respective
DL detectors as a function of SNR for correlated Rayleigh
channels, where we set $\rho=0.8$. Since MMNet fails to work in
the offline training mode for correlated channels, its results are
not included. Again, we see that OAMPNet suffer a substantial
performance degradation when an inaccurate noise variance is used.

\begin{figure}[t]
\setlength{\abovecaptionskip}{0pt}
\setlength{\belowcaptionskip}{0pt} \centering
\includegraphics[width=7.1cm]{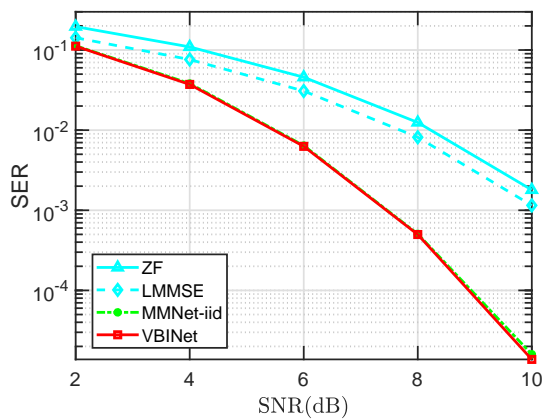}
\caption{SERs of different detectors vs. SNR on $128\times 64$
i.i.d. Gaussian channels.} \label{128size-IID-QPSK}
\end{figure}

\begin{figure}[t]
\setlength{\abovecaptionskip}{0pt}
\setlength{\belowcaptionskip}{0pt} \centering
\includegraphics[width=7.1cm]{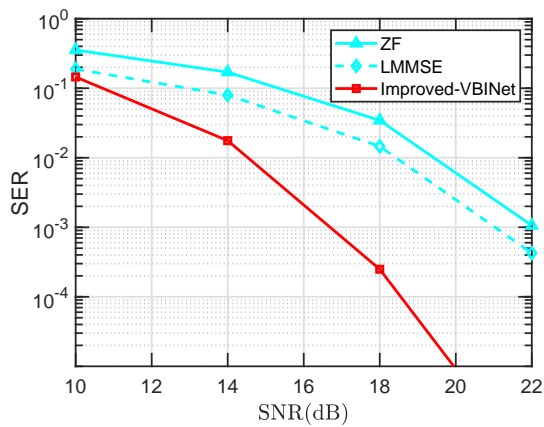}
\caption{SERs of different detectors vs. SNR on $128\times 64$
correlated Rayleigh channels.} \label{128size-correlated-QPSK}
\end{figure}

\subsubsection{Large-Scale MIMO Scenarios}
We examine the detection performance of our proposed method in the
large-scale MIMO scenario, in which the number of antennas at the
BS is up to 128. In our experiments, the QPSK modulation is
adopted, and we set $N_t=64$, $N_r=128$, $N_{\text{batch}}=500$,
and $N_{\text{iter}}=10^4$. For the correlated Rayleigh channel,
$\rho$ is set to $\rho=0.8$. Fig. \ref{128size-IID-QPSK} and Fig.
\ref{128size-correlated-QPSK} plot the SERs of respective
algorithms as a function of SNR for i.i.d. Gaussian and correlated
Rayleigh channels, respectively. Note that in the LE step of the
OAMPNet, $\boldsymbol{\hat{W}}_t$ is updated as
$\hat{\boldsymbol{W}}_t =v_t^2\boldsymbol{H}^H(v_t^2
\boldsymbol{H}\boldsymbol{H}^H+\boldsymbol{R}_{n})^{-1}$, which
involves computing the inverse of a matrix of size $N_r\times N_r$
at each layer. Hence for large-scale scenarios, training and test
of OAMPNet becomes computationally prohibitive. For this reason,
the results of OAMPNet are not included. Also, the results of
MMNet are not included for correlated Rayleigh channels as it aims
to learn an entire matrix $\boldsymbol{A}_t$ for each layer and
thus fails to work well in the offline training mode. Results in
Fig. \ref{128size-IID-QPSK} and Fig. \ref{128size-correlated-QPSK}
show that the proposed VBINet can achieve superior performance
while with a low computational complexity.

\section{Conclusions}   \label{section-conclusion}
In this paper, we proposed a model-driven DL network for MIMO
detection based on a variational Bayesian framework. Two networks
are respectively developed for i.i.d. Gaussian channels and
arbitrarily correlated channels. The proposed networks, referred
to as VBINet, have only a few learnable parameters and thus can be
efficiently trained with a moderate amount of training. Simulation
results show that the proposed detectors provide competitive
performance for both i.i.d. Gaussian channels and realistic 3GPP
MIMO channels. Moreover, the VBINet-based detectors can
automatically determine the noise variance, and achieve a
substantial performance improvement over existing DL-based
detectors such as OAMPNet and MMNet in the presence of noise
uncertainty.


\begin{thebibliography}{10}
\providecommand{\url}[1]{#1} \csname url@rmstyle\endcsname
\providecommand{\newblock}{\relax}
\providecommand{\bibinfo}[2]{#2}
\providecommand\BIBentrySTDinterwordspacing{\spaceskip=0pt\relax}
\providecommand\BIBentryALTinterwordstretchfactor{4}
\providecommand\BIBentryALTinterwordspacing{\spaceskip=\fontdimen2\font
plus \BIBentryALTinterwordstretchfactor\fontdimen3\font minus
  \fontdimen4\font\relax}
\providecommand\BIBforeignlanguage[2]{{%
\expandafter\ifx\csname l@#1\endcsname\relax
\typeout{** WARNING: IEEEtran.bst: No hyphenation pattern has been}%
\typeout{** loaded for the language `#1'. Using the pattern for}%
\typeout{** the default language instead.}%
\else \language=\csname l@#1\endcsname \fi #2}}

\bibitem{ZhangCui05}
Q.~Zhang, X.~Cui, and X.~Li, ``{Very tight capacity bounds for
MIMO-correlated
  Rayleigh-fading channels},'' \emph{IEEE Transactions on Wireless
  Communications}, vol.~4, no.~2, pp. 681--688, 2005.

\bibitem{ShahmansooriGarcia17}
A.~Shahmansoori, G.~E. Garcia, G.~Destino, G.~Seco-Granados, and
H.~Wymeersch,
  ``{Position and orientation estimation through millimeter-wave MIMO in 5G
  systems},'' \emph{IEEE Transactions on Wireless Communications}, vol.~17,
  no.~3, pp. 1822--1835, 2017.

\bibitem{LuLi14}
L.~Lu, G.~Y. Li, A.~L. Swindlehurst, A.~Ashikhmin, and R.~Zhang,
``{An overview
  of massive MIMO: Benefits and challenges},'' \emph{IEEE journal of selected
  topics in signal processing}, vol.~8, no.~5, pp. 742--758, 2014.

\bibitem{NgoLarsson13}
H.~Q. Ngo, E.~G. Larsson, and T.~L. Marzetta, ``Energy and
spectral efficiency
  of very large multiuser {MIMO} systems,'' \emph{IEEE Transactions on
  Communications}, vol.~61, no.~4, pp. 1436--1449, 2013.

\bibitem{KongXia16}
D.~Kong, X.-G. Xia, and T.~Jiang, ``{A differential QAM detection
in uplink
  massive MIMO systems},'' \emph{IEEE Transactions on Wireless Communications},
  vol.~15, no.~9, pp. 6371--6383, 2016.

\bibitem{tipping01}
M.~E. Tipping, ``{Sparse Bayesian learning and the relevance
vector machine},''
  \emph{Journal of machine learning research}, vol.~1, no. Jun, pp. 211--244,
  2001.

\bibitem{DonohoMaleki10}
D.~L. Donoho, A.~Maleki, and A.~Montanari, ``{Message passing
algorithms for
  compressed sensing: I. motivation and construction},'' in \emph{2010 IEEE
  information theory workshop on information theory (ITW 2010, Cairo)}.\hskip
  1em plus 0.5em minus 0.4em\relax IEEE, 2010, pp. 1--5.

\bibitem{ChenWipf16}
W.~Chen, D.~Wipf, Y.~Wang, Y.~Liu, and I.~J. Wassell,
``{Simultaneous Bayesian
  sparse approximation with structured sparse models},'' \emph{IEEE Trans.
  Signal Processing}, vol.~64, no.~23, pp. 6145--6159, 2016.

\bibitem{WuKuang14}
S.~Wu, L.~Kuang, Z.~Ni, J.~Lu, D.~Huang, and Q.~Guo,
``{Low-complexity
  iterative detection for large-scale multiuser MIMO-OFDM systems using
  approximate message passing},'' \emph{IEEE Journal of Selected Topics in
  Signal Processing}, vol.~8, no.~5, pp. 902--915, 2014.

\bibitem{Minka01}
T.~P. Minka, ``{A family of algorithms for approximate Bayesian
inference},''
  Ph.D. dissertation, Massachusetts Institute of Technology, 2001.

\bibitem{jaldenOttersten08}
J.~Jald{\'e}n and B.~Ottersten, ``The diversity order of the
semidefinite
  relaxation detector,'' \emph{IEEE Transactions on Information Theory},
  vol.~54, no.~4, pp. 1406--1422, 2008.

\bibitem{WangWenWang17}
T.~Wang, C.-K. Wen, H.~Wang, F.~Gao, T.~Jiang, and S.~Jin, ``{Deep
learning for
  wireless physical layer: Opportunities and challenges},'' \emph{China
  Communications}, vol.~14, no.~11, pp. 92--111, 2017.

\bibitem{WeiZhao20}
Y.~Wei, M.-M. Zhao, M.~Hong, M.-J. Zhao, and M.~Lei, ``{Learned
conjugate
  gradient descent network for massive MIMO detection},'' \emph{IEEE
  Transactions on Signal Processing}, vol.~68, pp. 6336--6349, 2020.

\bibitem{YeLi17}
H.~Ye, G.~Y. Li, and B.-H. Juang, ``{Power of deep learning for
channel
  estimation and signal detection in OFDM systems},'' \emph{IEEE Wireless
  Communications Letters}, vol.~7, no.~1, pp. 114--117, 2017.

\bibitem{BaiChen20}
Y.~Bai, W.~Chen, J.~Chen, and W.~Guo, ``{Deep learning methods for
solving
  linear inverse problems: Research directions and paradigms},'' \emph{Signal
  Processing}, p. 107729, 2020.

\bibitem{HuCai20}
Q.~Hu, Y.~Cai, Q.~Shi, K.~Xu, G.~Yu, and Z.~Ding, ``{Iterative
algorithm
  induced deep-unfolding neural networks: Precoding design for multiuser MIMO
  systems},'' \emph{IEEE Transactions on Wireless Communications}, vol.~20,
  no.~2, pp. 1394--1410, 2020.

\bibitem{MaGao21}
X.~Ma, Z.~Gao, F.~Gao, and M.~Di~Renzo, ``Model-driven deep
learning based
  channel estimation and feedback for millimeter-wave massive hybrid {MIMO}
  systems,'' \emph{IEEE Journal on Selected Areas in Communications}, 2021.

\bibitem{HuangLiu21}
Z.~Huang, A.~Liu, Y.~Cai, and M.-J. Zhao, ``Deep stochastic
optimization for
  algorithm-specific pilot design in massive {MIMO},'' in \emph{2021 IEEE
  Statistical Signal Processing Workshop (SSP)}.\hskip 1em plus 0.5em minus
  0.4em\relax IEEE, 2021, pp. 191--195.

\bibitem{BaekKwak19}
M.-S. Baek, S.~Kwak, J.-Y. Jung, H.~M. Kim, and D.-J. Choi,
``{Implementation
  methodologies of deep learning-based signal detection for conventional MIMO
  transmitters},'' \emph{IEEE Transactions on Broadcasting}, vol.~65, no.~3,
  pp. 636--642, 2019.

\bibitem{FarsadGoldsmith18}
N.~Farsad and A.~Goldsmith, ``Neural network detection of data
sequences in
  communication systems,'' \emph{IEEE Transactions on Signal Processing},
  vol.~66, no.~21, pp. 5663--5678, 2018.

\bibitem{UnShao19}
M.-W. Un, M.~Shao, W.-K. Ma, and P.~Ching, ``{Deep MIMO detection
using ADMM
  unfolding},'' in \emph{2019 IEEE Data Science Workshop (DSW)}.\hskip 1em plus
  0.5em minus 0.4em\relax IEEE, 2019, pp. 333--337.

\bibitem{HeWen20}
H.~He, C.-K. Wen, S.~Jin, and G.~Y. Li, ``{Model-driven deep
learning for MIMO
  detection},'' \emph{IEEE Transactions on Signal Processing}, vol.~68, pp.
  1702--1715, 2020.

\bibitem{KhaniAlizadeh20}
M.~Khani, M.~Alizadeh, J.~Hoydis, and P.~Fleming, ``{Adaptive
neural signal
  detection for massive MIMO},'' \emph{IEEE Transactions on Wireless
  Communications}, vol.~19, no.~8, pp. 5635--5648, 2020.

\bibitem{SamuelDiskin17}
N.~Samuel, T.~Diskin, and A.~Wiesel, ``{Deep MIMO detection},'' in
\emph{2017
  IEEE 18th International Workshop on Signal Processing Advances in Wireless
  Communications (SPAWC)}.\hskip 1em plus 0.5em minus 0.4em\relax IEEE, 2017,
  pp. 1--5.

\bibitem{MohammadMasouros20}
A.~Mohammad, C.~Masouros, and Y.~Andreopoulos,
``Complexity-scalable
  neural-network-based mimo detection with learnable weight scaling,''
  \emph{IEEE Transactions on Communications}, vol.~68, no.~10, pp. 6101--6113,
  2020.

\bibitem{AguerriZaidi21}
I.~E. Aguerri and A.~Zaidi, ``Distributed variational
representation
  learning,'' \emph{IEEE Transactions on Pattern Analysis and Machine
  Intelligence}, vol.~43, no.~1, pp. 120--138, 2021.

\bibitem{ZaidiAguerri20}
A.~Zaidi and I.~E. Aguerri, ``Distributed deep variational
information
  bottleneck,'' in \emph{2020 IEEE 21st International Workshop on Signal
  Processing Advances in Wireless Communications (SPAWC)}, Atlanta, GA, USA,
  2020.

\bibitem{DuanYang17}
H.~Duan, L.~Yang, J.~Fang, and H.~Li, ``Fast inverse-free sparse
bayesian
  learning via relaxed evidence lower bound maximization,'' \emph{IEEE Signal
  Processing Letters}, vol.~24, no.~6, pp. 774--778, 2017.

\bibitem{ZhuMurch02}
X.~Zhu and R.~D. Murch, ``{Performance analysis of maximum
likelihood detection
  in a MIMO antenna system},'' \emph{IEEE Transactions on Communications},
  vol.~50, no.~2, pp. 187--191, 2002.

\bibitem{MaPing17}
J.~Ma and L.~Ping, ``Orthogonal amp,'' \emph{IEEE Access}, vol.~5,
pp.
  2020--2033, 2017.

\bibitem{TzikasLikas08}
D.~G. Tzikas, A.~C. Likas, and N.~P. Galatsanos, ``The variational
  approximation for {Bayesian} inference,'' \emph{IEEE Signal Processing
  Magazine}, vol.~25, no.~6, pp. 131--146, 2008.

\bibitem{SunBabu16}
Y.~Sun, P.~Babu, and D.~P. Palomar, ``Majorization-minimization
algorithms in
  signal processing, communications, and machine learning,'' \emph{IEEE
  Transactions on Signal Processing}, vol.~65, no.~3, pp. 794--816, 2016.

\bibitem{MagnusNeudecker19}
J.~R. Magnus and H.~Neudecker, \emph{Matrix differential calculus
with
  applications in statistics and econometrics}.\hskip 1em plus 0.5em minus
  0.4em\relax John Wiley \& Sons, 2019.

\bibitem{RanganSchniter19}
S.~Rangan, P.~Schniter, A.~K. Fletcher, and S.~Sarkar, ``On the
convergence of
  approximate message passing with arbitrary matrices,'' \emph{IEEE
  Transactions on Information Theory}, vol.~65, no.~9, pp. 5339--5351, 2019.

\bibitem{ChenWang00}
R.~Chen, X.~Wang, and J.~S. Liu, ``{Adaptive joint detection and
decoding in
  flat-fading channels via mixture Kalman filtering},'' \emph{IEEE Transactions
  on Information Theory}, vol.~46, no.~6, pp. 2079--2094, 2000.

\bibitem{AbadiAgarwal16}
M.~Abadi, A.~Agarwal, P.~Barham, E.~Brevdo, Z.~Chen, C.~Citro,
G.~S. Corrado,
  A.~Davis, J.~Dean, M.~Devin, \emph{et~al.}, ``{Tensorflow: Large-scale
  machine learning on heterogeneous distributed systems},'' \emph{arXiv
  preprint arXiv:1603.04467}, 2016.

\bibitem{PaszkeGross17}
A.~Paszke, S.~Gross, S.~Chintala, G.~Chanan, E.~Yang, Z.~DeVito,
Z.~Lin,
  A.~Desmaison, L.~Antiga, and A.~Lerer, ``Automatic differentiation in
  pytorch,'' 2017.

\bibitem{KingmaBa14}
D.~P. Kingma and J.~Ba, ``{Adam: A method for stochastic
optimization},''
  \emph{arXiv preprint arXiv:1412.6980}, 2014.

\bibitem{optimizationinc15}
\BIBentryALTinterwordspacing {Gurobi Optimization. Inc. },
``Gurobi optimizer reference manual,'' 2021.
  [Online]. Available: \url{https://www. gurobi. com}
\BIBentrySTDinterwordspacing

\bibitem{Loyka01}
S.~L. Loyka, ``{Channel capacity of MIMO architecture using the
exponential
  correlation matrix},'' \emph{IEEE Communications letters}, vol.~5, no.~9, pp.
  369--371, 2001.

\bibitem{JaeckelRaschkowski14}
S.~Jaeckel, L.~Raschkowski, K.~B{\"o}rner, and L.~Thiele,
``{QuaDRiGa: A 3-D
  multi-cell channel model with time evolution for enabling virtual field
  trials},'' \emph{IEEE Transactions on Antennas and Propagation}, vol.~62,
  no.~6, pp. 3242--3256, 2014.

\bibitem{3GPP}
\BIBentryALTinterwordspacing {"3GPP TR 38.901 V16.1.0"},
\emph{{Technical Specification Group Radio Access
  Network; Study on channel model for frequencies from 0.5 to 100 GHz (Release
  14)}}, 2019. [Online]. Available:
  \url{https://www.3gpp.org/ftp/Specs/archive/}
\BIBentrySTDinterwordspacing

\end{thebibliography}

\end{document}